\documentclass[journal]{IEEEtran}

\usepackage{algorithmic}
\usepackage{algorithm}
\usepackage{array}

\usepackage{subfigure}
\usepackage{amsmath,amsfonts}
\usepackage{amsthm}

\usepackage{textcomp}
\usepackage{stfloats}
\usepackage{url}
\usepackage{verbatim}
\usepackage{graphicx}
\usepackage{cite}
\usepackage{bm}
\usepackage{soul}
\usepackage{arydshln}
\usepackage{booktabs}
\usepackage{pifont}
\usepackage{amssymb}
\usepackage{tabularx}
\usepackage{microtype}
\usepackage{stfloats}
\usepackage{epstopdf}
\usepackage{makecell}
\usepackage{diagbox}

\allowdisplaybreaks[4]
\usepackage[implicit=false]{hyperref}
\usepackage{tikz,xcolor}

\newtheorem{lemma}{Lemma}

\newtheorem{corollary}{Corollary}
\newtheorem{proposition}{Proposition}

\begin{document}

\title{Networked Tracking of Multiple Moving Targets in 6G Network}
\author{Yanmo Hu,~\IEEEmembership{Member,~IEEE},
Weifeng Zhu,~\IEEEmembership{Member,~IEEE},
Chenshu Wu,~\IEEEmembership{Senior~Member,~IEEE},
Shuowen Zhang,~\IEEEmembership{Senior~Member,~IEEE},
J. Andrew Zhang,~\IEEEmembership{Senior~Member,~IEEE},
and
Liang Liu,~\IEEEmembership{Fellow,~IEEE}\vspace{-0.25em}
\thanks{
Y. Hu, W. Zhu, S. Zhang, and L. Liu
are with the
Department of Electrical and Electronic Engineering, The Hong Kong Polytechnic
University, Hong Kong, SAR, China
(e-mail: \{yanmo.hu, eee-wf.zhu, shuowen.zhang, liang-eie.liu\}@polyu.edu.hk).

C. Wu is with the Department of Computer
Science, The University of Hong Kong, Hong Kong, SAR, China
(e-mail: chenshu@cs.hku.hk).

J. A. Zhang is with Global Big Data Technologies Centre, the University of Technology Sydney,
Sydney,
Australia
(e-mail: Andrew.Zhang@uts.edu.au).

}}
\maketitle

\begin{abstract}

This paper considers a networked tracking architecture in
6G integrated sensing and communication (ISAC) systems,
where multiple base stations (BSs)
cooperatively transmit radio signals
and process received echo signals
to track multiple moving targets.
Compared to the single-BS counterpart,
networked tracking allows the moving targets to be associated with different BSs over time such that the wireless resources can be dynamically allocated among BSs based on target locations.
However,
networked tracking imposes new challenges for algorithm design and resource allocation.
In this paper,
we first design the networked Kalman Filter (NKF) that is suitable for multi-BS based tracking,
then characterize the posterior Cram\'{e}r-Rao bound (PCRB)
under this NKF, and last design the beamforming vectors of all the BSs to minimize the tracking PCRB.
Numerical results show that our dynamic beamforming design can properly associate the targets to the suitable BSs at various sensing blocks and reduce the tracking mean-squared error (MSE).

\end{abstract}


\section{Introduction}\label{Section_Introduction}

\subsection{Motivations and Backgrounds}

\IEEEPARstart{I}{ntegrated}
sensing and communications (ISAC)
\cite{9737357, 9540344, 10663814}
have been recognized as one of the six key
scenarios for 6G development \cite{9705498}
and
have garnered significant attention over the past decade.
Currently, research on single-node sensing technologies,
encompassing both monostatic \cite{tongkai_2026_ICASSP, 234560} and bistatic \cite{10637442} sensing,
has achieved significant advancements,
providing a solid foundation for the development
of ISAC systems.
However, the limitations of single-node sensing manifest in several key aspects: a restricted sensing range, incomplete parameter estimation (e.g., the inability to estimate tangential velocity), and a lack of adaptive flexibility, etc.
Driven by ongoing research efforts,
networked sensing \cite{10462908, wang2024fue, 10557715},
which enables sensing across multiple cooperated nodes from diverse observation perspectives,
has increasingly attracted attention.
By leveraging the synergy of multiple nodes,
networked sensing offers a more comprehensive and robust framework for addressing complex sensing tasks in various environments.

Networked sensing can leverage the ubiquitous deployment of base stations (BSs) and/or user equipments (UEs),
making it particularly well-suited for ISAC applications.
Current research primarily focuses on three key scenarios:
UE-BS-based cooperative sensing \cite{10636720, 10462908},
UE-based networked sensing \cite{wang2024fue},
and
BS-based networked sensing \cite{10557715, 10226276, 10684491, 11288092, gao2025covariancebasedimagingmultiviewfusion}.
The UE-BS-based cooperative sensing combines the strengths of both uplink (UL) and downlink (DL) sensing.
Study in \cite{10462908}
thoroughly investigates
the sensing accuracy of the cooperative approach compared to non-cooperative sensing.
Results demonstrate that cooperative sensing significantly enhances sensing accuracy, while also achieving higher communication reliability.
Furthermore,
in \cite{10636720},
a generalized
multi-BS and multi-UE
configuration is proposed for 4D environmental reconstruction,
successfully implementing multi-level data fusion
through a deep learning approach.
The
UE-based
networked sensing
relies exclusively on UL signals transmitted from multiple UEs units
to the receiver of a specific BS,
significantly reducing system complexity.
However, this sensing configuration suffers from
challenges such as
the clock asynchronism and requires efficient
resource allocation.
A representative study in \cite{wang2024fue}
effectively addresses these issues by
identifying optimal and worst-case resource allocation strategies for each sensing block.
The
BS-based networked sensing utilizes DL signals across multiple BSs,
allowing for the integration of a large volume of DL signals with known modulated symbols.
A critical challenge of this scenario is the fusion of DL signals.
In \cite{10557715},
a data-level fusion approach is proposed
to jointly estimate target parameters,
supported by impressive experimental results,
while \cite{10226276}
introduces a symbol-level fusion approach, outperforming
both data-level multi-BS sensing and single-BS sensing in terms of the accuracy of location and velocity estimation.
However,
current advanced networked sensing schemes localize targets of interest within individual sensing blocks.
Due to the temporal correlation between adjacent blocks,
sensing results are inherently interdependent.
Treating these results as independent may lead to performance degradation,
thereby limiting the overall sensing accuracy.
An effective approach to leverage the coupled data
across multiple blocks
to achieve better fusion performance
is \emph{target tracking}.

Target tracking,
which facilitates parameter estimation across multiple sensing blocks,
serves as an effective method
capable of achieving high-precision sensing tasks
by continuously updating and refining the target's state based on incoming measurements.
At its core,
tracking process relies on two fundamental models:
the target motion model and the measurement model \cite{book_tracking1}.
The target motion model,
also referred to as the dynamic model,
plays a critical role in predicting the future state.
This model is heavily influenced by system architectures,
sensing environments,
and target characteristics.
Consequently,
typical target motion models, including
the constant velocity (CV),
constant acceleration (CA),
and coordinated turn (CT) models \cite{251886, 4102704},
are developed to tailor to the specific practical scenario.
The measurement model
for target tracking
describes how sensor observations relate to
the targets' state.
When the relationship between the measurement and the targets' state is linear,
the conventional Kalman filter is well-suited for the tracking.
However,
in cases where this relationship is nonlinear,
often encountered in complex scenarios,
the tracking process requires advanced filtering techniques,
including the extended Kalman filter (EKF) \cite{book_tracking1}, unscented Kalman filter (UKF) \cite{1271397}, and quadrature Kalman filter (QKF) \cite{855552}, etc.
Since the target tracking
is able
to mitigate measurement errors,
it is widely employed in various sensing schemes.

Because of the high-precision advantages,
target tracking
has garnered significant research interest.
Existing studies can be broadly categorized into
single-node tracking  \cite{9679386, 5592379},
single-target networked tracking \cite{9896671, 10214383}
and multi-target networked tracking \cite{6095653, 9119156, 6960029}.
As the most fundamental category,
single-node tracking is investigated
by
\cite{9679386},
which introduces a tracking-based algorithm for WiFi-based uplink sensing.
By employing a Kalman filter on position and Doppler measurements,
it achieves high-precision localization
and showcases the potential of
uplink tracking in ISAC systems.
Networked tracking schemes, in contrast, offer greater system gain.
For instance,
\cite{9896671} designs
a UAV-enabled networked sensing framework to reduce velocity estimation errors for a single target,
demonstrating superior performance over single-BS methods.
Beam optimization is another key focus in this category.
In \cite{10214383},
an EKF-based algorithm is proposed for angle prediction,
which presents a fine beam-tracking method
to enhance sensing accuracy.
Building on these foundations,
multi-target networked tracking has emerged as a more practical and increasingly active area,
where beam optimization remains a central research challenge.
A commonly adopted approach of these studies is to utilize the posterior Cram\'{e}r-Rao bound (PCRB) \cite{668800}
as the optimization criterion for designing optimal system parameters.
For instance, \cite{6095653} presents an optimization method for a cognitive radar network.
This method designs the antenna scheduling and power allocation for each radar to track multiple targets in multipath scenarios.
Similarly,
the studies in \cite{9119156, 6960029}
leverage MIMO techniques to optimize transmit signals for multi-target tracking,
enabling the generation of multiple beams
for simultaneous tracking.
However,
these studies rely on
multiple beams
with
each exclusively directing at a single target without interference from others.
Due to the limitation of the beam resources in communication systems,
it is infeasible for the BS to direct its beams
toward all targets simultaneously across all sensing blocks.

\begin{figure}[!t]
\centering
{\includegraphics[width=2.9in]{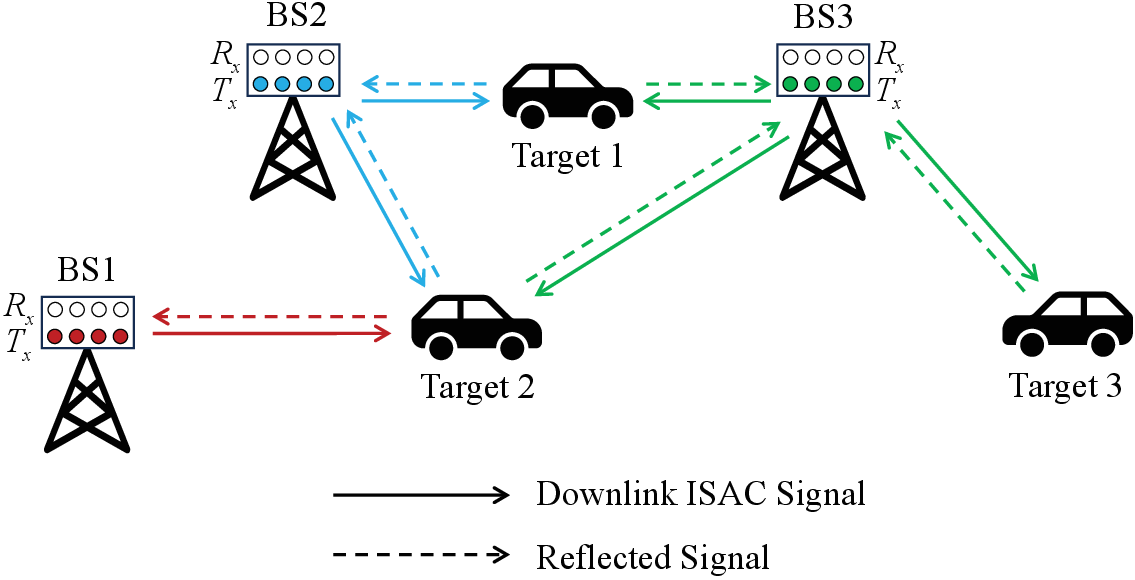}}
\caption{System model of our considered networked tracking architecture,
where
the BSs transmit downlink signals to track multiple moving targets.
In particular,
different BSs have different target association strategies.}
\label{Fig_picture}
\end{figure}

\subsection{Contributions and Organizations}

In this paper,
to efficiently achieve tracking in multi-target networked systems,
we realize
an optimal beam allocation strategy
which can
achieve the target association dynamically at each BS.
Specifically,
as depicted in Fig. \ref{Fig_picture},
different BSs have different target association strategies.
For instance,
the long distance between BS 1 and Targets 1 and 3
results in BS 1 serving only Target 2.
This association strategy
prevents the inefficient allocation of beams to distant links.
To achieve such a function,
we develop a networked tracking
scheme along with
beamforming optimization
in the considered ISAC networked systems.
We focus on two main problems.
For the networked sensing scheme with given beamformers,
a networked tracking model is constructed to
characterize the temporal and geometrical correlations of the estimated parameters.
To leverage these multidimensional features, a Networked Kalman Filter (NKF) is proposed to achieve optimal data fusion.
Then,
beamforming optimization is formulated
to dynamically optimize the beamformers at each sensing block to achieve
beam association and
optimal tracking accuracy.
Our main contributions are highlighted as follows.
\begin{itemize}

\item
We propose a multi-target tracking framework specifically designed for multi-BS networked sensing systems.
The framework is realized through a novel NKF that achieves high-precision sensing by exploiting both temporal and geometrical correlations in the estimates.
The devised NKF algorithm effectively performs multidimensional data fusion,
enabling the simultaneous high-precision tracking of each target's location, velocity, and equivalent radar cross section (RCS).

\item

We provide the PCRB for the proposed multi-target networked tracking model
for performance analysis.
Specifically,
the PCRB at each tracking step is calculated recursively
from the bounds of the preceding time step.
Furthermore,
we provide a detailed analysis of the PCRB for the networked scheme,
which yields critical insights into how key parameters influence the ultimate tracking performance.

\item
We propose
a beamforming optimization algorithm
to
dynamically optimize the
beamformer of each BS in every sensing block
using a PCRB minimization criterion.
However,
this optimization faces a fundamental challenge:
due to target motion,
beamforming for the next sensing block
should be optimized in advance,
since the future estimated parameters are not yet available.
To address this problem,
we
leverage the predicted target parameters for beamforming optimization
using the semidefinite program (SDP).
Simulation results confirm
that the proposed algorithm enables dynamic
target association at each BS,
thereby enhancing networked tracking flexibility in multi-target scenarios.

\end{itemize}

The rest of this paper is organized as follows.
Section \ref{Section_Signal_Model}
introduces the signal model for networked sensing.
Section \ref{Section_Tracking_Model}
provides the
tracking model
and the nonlinear filtering algorithm.
The
performance bounds of our scheme is given in
Section \ref{Section_performance_PCRB_analysis},
followed by a detailed analysis of the sensing metrics.
Section \ref{Section_tracking_optimizatio}
presents the beamforming optimization algorithm.
Simulation results are shown in Section \ref{Section_simulation_results}
and conclusions are drawn in Section \ref{Section_conclusion}.

$Notation$:
$\odot $, $\otimes $ and $\oplus$ denote the Hadamard, Kronecker and Khatri-Rao products, respectively;
${{\mathbf{1}}_{K_1 \times K_2}}\in {{\mathbb{Z}}^{K_1 \times K_2}}$ denotes the matrix with all elements being 1;
${{\mathbf{I}}_{M}}\in {{\mathbb{Z}}^{M\times M}}$ denotes identity matrix;
${{\mathbf{A}}^{T}}$, ${{\mathbf{A}}^{*}}$, and ${{\mathbf{A}}^{H}}$ represent the transpose, conjugate, and conjugate transpose of $\mathbf{A}$, respectively;
$\text{blkdiag}\left\{\cdot\right\}$
denotes the block diagonalization operator;
$\mathbf{E}_{i, j}$ denotes the matrix
where the $\left(i, j\right)$-th element is 1,
while all other elements are 0;
$\mathbf{e}_{M, m} \in \mathbf{Z}^{M \times 1}$,
$m = 0, 1, \cdots, M - 1$,
denotes the vector where the $m$-th element is 1, while all other elements are 0.

\section{System Model for Networked Tracking}\label{Section_Signal_Model}
\subsection{Networked Tracking Scheme}\label{Section_sensing_scheme_introduction}

We consider a MIMO-OFDM-based ISAC
cellular system,
as depicted in Fig. \ref{Fig_picture},
where $K$ base stations (BSs)
transmit downlink beams to communicate with $R$ users and track
$Q$ moving targets using the echo signals simultaneously.
The 2D coordinate of BS $k$ is denoted as
$\mathbf{p}_k = \left[p_{x,k}, p_{y,k}\right]^T$
in meters, $\forall k$.
Furthermore,
it is assumed that each
BS is equipped with $N_T$ transmit antennas and $N_R$
receive antennas and operates in full-duplex mode.
In practice, the transmit and receive antenna array at each BS should be sufficiently separated such that the self-interference
can be largely suppressed together with RF and digital interference cancellation techniques
\cite{9724258}.

Moreover, let $\mathcal{M} = \left\{0,1, \cdots,M - 1\right\}$ denote the set of OFDM sub-carriers, where $M$ is the number of sub-carriers.
To mitigate the interference from signals transmitted by other BSs
during the tracking process,
we follow the orthogonal sub-carrier allocation scheme \cite{10436573}.
Specifically, let
$\mathcal{M}_k \subseteq \mathcal{M}$
and
$M_k$
denote the set and the number of sub-carriers allocated to BS $k$,
respectively, $\forall k$.
Then, it is assumed that
${\mathcal{M}_{k}}\cap {\mathcal{M}_{\bar{k}}}=\varnothing$ for any ${k} \ne {\bar{k}}$,
and
$\bigcup _{k=0}^{K-1}{\mathcal{M}_{k}}=\mathcal{M}$.
Then, each BS will only transmit radio signals and receive echo signals at its assigned sub-carriers to track the moving targets.

\subsection{Signal Model for Networked Tracking}\label{Section_dlkjflskdjflksjdf}

The OFDM symbols used for tracking the moving targets
are divided into multiple sensing blocks,
each comprising $I$ OFDM symbols.
Within a given sensing block,
target location and velocity are assumed to be constant across different symbols,
which is a common assumption known as the ``stop-and-go'' approximation \cite{5420035}.
Under this approximation, at sensing block $n$,
the 2D coordinate of target $q$ is denoted as
$\left(x_{x,n,q}, x_{y,n,q}\right)$ in meters,
and the velocity of target $q$ along the $x$-axis and $y$-axis
is denoted as $v_{x,n,q}$ and $v_{y,n,q}$ in meters/second, respectively, $\forall q$.
Our goal is to estimate
$x_{x,n,q}$,
$x_{y,n,q}$,
$v_{x,n,q}$,
and
$v_{y,n,q}$, $\forall n$, for the tracking purpose.

Specifically, during sensing block $n$,
the frequency-domain downlink OFDM signal emitted by
the transmit antenna array of the $k$-th BS
at the
$m$-th sub-carrier
of
the
$i$-th OFDM symbol duration
is given by
\begin{align}\label{formula_signal_transmitting}
{{\mathbf{s}}_{n,k,i, m}}& ={\sum\limits_{\ell=0}^{L_{n,k,m}-1}{{{\mathbf{b}}_{n,k,m,\ell }}{{c}_{n,k,\ell ,i,m}}}}={{{{\mathbf{{B}}}}_{n,k,m}}{{\mathbf{c}}_{n,k,i,m}}},
\end{align}
where
$L_{n,k,m}$ denotes the number of beams
transmitted by the $k$-th BS of sensing block $n$ at the $m$-th subcarrier,
${{\mathbf{b}}_{n,{{k}}, m,\ell}}$
is the $\ell$-th transmit beamformer of the $k$-th BS of sensing block $n$ at the $m$-th subcarrier,
${{{\mathbf{{B}}}}_{n,k,m}}=\left[ {{\mathbf{b}}_{n,k,m,0}},\cdots ,{{\mathbf{b}}_{n,k,m,{L_{n,k,m}}-1}} \right]\in {{\mathbb{C}}^{{{N}_{T}}\times {L_{n,k,m}}}}$
is the matrix
containing all the beamformers of the $k$-th BS during sensing block $n$ at the $m$-th subcarrier,
and
${{\mathbf{c}}_{n,k,i,m}}={{\left[ {{c}_{n,k,0,i,m}},\cdots ,{{c}_{n,k,{L_{n,k,m}}-1,i,m}} \right]}^{T}}\in {{\mathbb{C}}^{{L_{n,k,m}}\times 1}}$
is the modulated symbol vector
with
${{c}_{n,k,\ell,i,m}}$
denoting the modulated unit-power symbol
at the $\ell$-th transmit beam
of the $k$-th BS at the $m$-th sub-carrier during the $i$-th OFDM
symbol duration of the $n$-th sensing block.
Note that
${\mathbf{{c}}_{n,k,i,m}} = 0$
if $m \notin \mathcal{M}_k$
and
${\mathbf{{c}}_{n,k,i,m}}\sim \mathcal{C}\mathcal{N}\left( \mathbf{0}_{L_{n,k,m} \times 1}, \mathbf{I}_{L_{n,k,m}} \right)$
if $m \in \mathcal{M}_k$.
Moreover,
the covariance matrix of the signals
emitted by the transmit antenna array at BS $k$ of sensing block $n$ at the $m$-th subcarrier
can be expressed as
$\mathbf{R}_{B, n,k,m}=\mathbb{E}\left\{ {{{\mathbf{{B}}}}_{n,k,m}}{{\mathbf{c}}_{n,k,i,m}}\mathbf{c}_{n,k,i,m}^{H}\mathbf{{B}}_{n,k,m}^{H} \right\}={{{\mathbf{{B}}}}_{n,k,m}}\mathbf{{B}}_{n,k,m}^{H}$,
which follows the rank constraint of $\text{rank}\left\{ \mathbf{R}_{B,n,k,m} \right\} \le L_{n, k, m}$.

Under the considered model, all the BSs transmit the above wireless signals
simultaneously in the downlink and receive the echo signals from the targets.
Let
${{\theta }_{n,k,q}}$,
${{\varphi }_{n,k,q}}$,
${{\tau }_{n,k,q}}$,
and
${{f}_{n,k,q}}$
denote the AOA, AOD, delay, and Doppler frequency between the $q$-th target and the $k$-th BS
at sensing block $n$, respectively,
$\forall n, k, q$.
These parameters are characterized as
\begin{align}\label{parameter_transform}
\nonumber {{\theta }_{n,{{k}},q}}=& \, \arctan \frac{{{x}_{x,n,q}}-{{p}_{x,{{k}}}}}{{{x}_{y,n,q}}-{{p}_{y,{{k}}}}}+\Delta {{\theta }_{{{k}}}}, \\
\nonumber  {{\varphi }_{n,{{k}},q}} = & \,{{\theta }_{n,{{k}},q}} - \Delta {{\theta }_{{{k}}}} + \Delta {{\varphi }_{{{k}}}}, \\
\nonumber {{\tau }_{n,{{k}},q}}=& \, \frac{2}{c}{{D}_{n,{{k}_{t}},q}}, \\
{{f}_{n,{{k}},q}}= & \, \frac{2\left[ {{v}_{x,n,q}},{{v}_{y,n,q}} \right]}{\lambda {{{D}_{n,{{k}},q}}}}\left({{{\mathbf{p}}_{{{k}}}}-{{\left[ {{x}_{x,n,q}},{{x}_{y,n,q}} \right]}^{T}}}\right),
\end{align}
where
$\lambda$ denotes the wavelength,
$c$ represents the speed of the light,
$\Delta {{\theta }_{k}}$
and
$\Delta {{\varphi }_{k}}$
are the
inclination angle of the receive and the transmit arrays,
respectively,
and
${{D}_{n,k,q}}\triangleq {{\left\| {{\mathbf{p}}_{k}}-{{\left[ {{x}_{x,n,q}},{{x}_{y,n,q}} \right]}^{T}} \right\|}_{2}}$
denotes the range from the $k$-th BS to the $q$-th target at sensing block $n$.
Employing the above target parameters, during the $i$-th OFDM symbol of sensing block $n$,
the echo signal obtained
from the receive antenna array of the $k$-th BS
at sub-carrier ${m} \in \mathcal{M}_k$
is given by
\begin{align}\label{formula_receiving_MIMO_all}
\nonumber {{\mathbf{y}}_{n,k,i,{{m}}}}= & {\sum\limits_{q=0}^{Q-1}{\underbrace{{{\mathbf{a}}_{r}}\left( {{\theta }_{n,k,q}} \right)}_{\text{AOA term}}{{\xi }_{n,k,q}}\underbrace{\mathbf{a}_{t}^{H}\left( {{\varphi }_{n,{{k}},q}} \right)}_{\text{AOD term}}{{\mathbf{B}}_{n,k,m}}{{\mathbf{c}}_{n,k,i,{{m}}}}}} \\
 & \times \underbrace{{{e}^{j2\pi {{f}_{n,k,q}}i{{T}_{0}}}}}_{\text{Doppler term}}\underbrace{{{e}^{j2\pi {{m}}\Delta f{{\tau }_{n,k,q}}}}}_{\text{Delay term}}+{{\mathbf{z}}_{n,k,i,{{m}}}}, \ {m} \in \mathcal{M}_k,
\end{align}
where
${{\mathbf{a}}_{r}}\left( {{\theta }_{n,k,q}} \right)=\left[ 1,\cdots ,{{e}^{j\frac{2\pi }{\lambda }\left( {{N}_{R}}-1 \right)d\sin {{\theta }_{n,k,q}}}} \right]\in {{\mathbb{C}}^{{{N}_{R}}\times 1}}$
denotes
the receive steering vector for ${{\theta }_{n,k,q}}$,
${{\mathbf{a}}_{t}}\left( {{\varphi }_{n,k,q}} \right)=\left[ 1,\cdots ,{{e}^{j\frac{2\pi }{\lambda }\left( {{N}_{T}}-1 \right)d\sin {{\varphi }_{n,k,q}}}} \right]\in {{\mathbb{C}}^{{{N}_{T}}\times 1}}$
represents the transmit steering vector for ${{\varphi }_{n,k,q}}$,
$T_0$ denotes the OFDM symbol interval,
and
$\mathbf{z}_{n,{{k}},i,m} \in \mathbb{C}^{N_R \times 1}$ is the additive white Gaussian noise with
$\mathbf{z}_{n,{{k}},i,m}\sim \mathcal{C}\mathcal{N}\left( \mathbf{0}_{N_R \times 1},\sigma _{z}^{2} \mathbf{I}_{N_R} \right)$.
Additionally,
$\xi_{n,{{k}},q}$
is the complex coefficient containing the propagation attenuation and
radar cross-section (RCS) of the $q$-th target at the $k$-th BS during sensing block $n$,
and
its precise expression is given by
\begin{align}\label{RCS_fomula}
{{\xi }_{n,{{k}},q}}= \frac{{\sigma }_{R,n,{{k}},q}}{D_{n,{{k}},q}^{2}}\sqrt{ \frac{P_k G_k^2 \lambda^2}{\left(4 \pi\right)^3 }},
\end{align}
where $P_k$ and $G_k$ denote the
transmit power and the antenna gain of the $k$-th BS,
respectively,
both of which are precisely known,
and ${{\sigma }_{R,n,{{k}},q}} \in \mathbb{C}$
is
the equivalent RCS (ERCS).

\begin{figure}[!t]
\centering
{\includegraphics[width=3.4in]{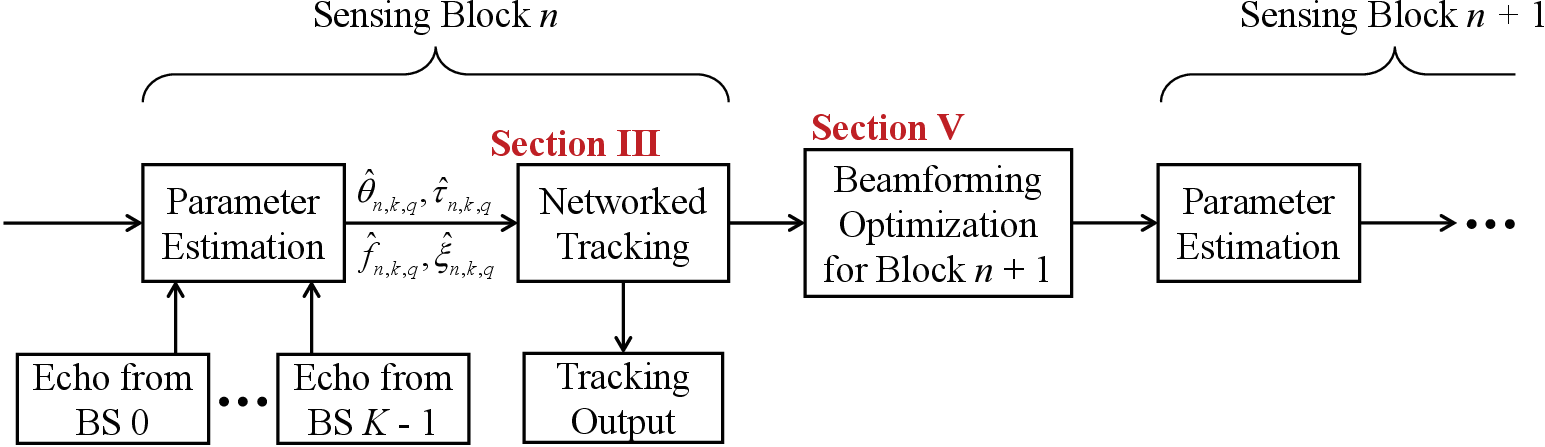}}
\caption{Diagram of the proposed two-phase networked sensing scheme,
where the networked tracking is performed after the parameter estimation,
followed by beamforming optimization for the next sensing block.}
\label{Fig_process}
\end{figure}

\subsection{Parameter Estimation}

Then,
as depicted in Fig. \ref{Fig_process},
using the received signals given in (\ref{formula_receiving_MIMO_all}),
each BS estimates the
sensing parameters at each sensing block.
Since the single-node non-tracking
parameter
estimation in ISAC is quite mature \cite{11079818, 10529184},
we can directly
apply typical sensing algorithms,
such as the FFT-based and MLE-based algorithms,
to (\ref{formula_receiving_MIMO_all}),
and
obtain the estimates
${\hat{\theta }_{n,k,q}}$,
${\hat{\tau }_{n,k ,q}}$,
${\hat{f}_{n ,k,q}}$,
and
${\hat{\xi}_{n,k,q}}$, $\forall n,k,q$.
Note that
directly estimating ${{\varphi }_{n,k,q}}$ from the received signals becomes impossible
because
$\mathbf{a}_{t}^{H}\left( {{\varphi }_{n,{{k}},q}} \right){{\mathbf{{B}}}_{n,{{k}},m}}$ is contaminated by ${{\xi }_{n,{{k}},q}}$, as evident from (\ref{formula_receiving_MIMO_all}).
Referring to (\ref{parameter_transform}),
a feasible approach for obtaining ${{\varphi }_{n,k,q}}$ is to derive it from ${{\theta }_{n,k,q}}$.
Because the estimate ${\hat{\varphi }_{n,k,q}}$
is determined by ${\hat{\theta }_{n,k,q}}$,
it cannot provide extra localization information,
so we will adopt
${\hat{\theta }_{n,k,q}}$,
${\hat{\tau }_{n,k ,q}}$,
${\hat{f}_{n ,k,q}}$,
and
${\hat{\xi}_{n,k,q}}$
for subsequent data fusion.

Based on the above estimated parameters,
in this paper,
we aim to achieve high-precision networked localization
through the following three steps.
First,
given the transmit beamforming vectors at each BS,
we propose a networked Kalman filter (NKF)
algorithm to perform
multidimensional data fusion and obtain high-precision tracking
results.
Next,
we characterize the PCRB of the constructed networked tracking model.
This bound quantifies the tracking performance as a function of the beamforming vectors.
Finally,
we design the beamforming vectors for each BS at each sensing block
to minimize the PCRB for multi-target networked sensing.

\section{Networked Tracking}\label{Section_Tracking_Model}

In this section,
we aim to
achieve the multidimensional data fusion for high-precision localization tasks.
Due to target motion and BS-based cooperative sensing,
the estimated parameters,
i.e.,
${\hat{\theta }_{n,k,q}}$,
${\hat{\tau }_{n,k ,q}}$,
${\hat{f}_{n ,k,q}}$,
and
${\hat{\xi}_{n,k,q}}$,
are inherently interdependent
across multiple sensing blocks, i.e., \emph{temporal correlation},
and
among all BSs, i.e., \emph{geometrical correlation}.
To leverage both correlations,
specifically,
the temporal correlation
will be revealed in
our proposed networked motion model,
and
the geometrical
correlation
will be unveiled in
our proposed networked measurement model.
Both models constitute the networked tracking model.
Subsequently,
based on the networked tracking model,
the NKF algorithm is designed
to achieve the data fusion.

\subsection{Networked Motion Model for Characterizing Temporal Correlation}

Our networked motion model
will mainly characterize the temporal correlation of the estimated parameters
across multiple sensing blocks.
As is typical,
we
assume that each target conforms to a constant velocity model\footnote{It is worth noting that our tracking model can also be extended to more general cases
involving accelerated speed. For conciseness, we exemplify the constant velocity model
in this paper.}.
This modeling assumption enables efficient characterization and utilization of temporal correlations
to
establish explicit relationships between target parameters in adjacent sensing blocks.

Specifically,
suppose that
the state vector of all targets at block $n$ is given by
${{\bm{{\kappa} }}_{A,n}}={{\left[ {\mathbf{x}_{x, n}^{T}},{\mathbf{x}_{y, n}^{T}},{\mathbf{v}_{x, n}^{T}},{\mathbf{v}_{y, n}^{T}} \right]}^{T}}\in {{\mathbb{R}}^{4Q \times 1}}$,
where
${\mathbf{x}_{x, n}}={{\left[ x_{x, n,0},\cdots ,x_{x,n,Q-1} \right]}^{T}} \in \mathbb{R}^{Q \times 1}$ and
${\mathbf{x}_{y, n}}={{\left[ x_{y,n,0},\cdots ,x_{y,n,Q-1} \right]}^{T}} \in \mathbb{R}^{Q\times 1}$ represent the position vectors, and
${\mathbf{v}_{x, n}}={{\left[ v_{x,n,0},\cdots ,v_{x,n,Q-1} \right]}^{T}} \in \mathbb{R}^{Q\times 1}$ and
${\mathbf{v}_{y, n}}={{\left[ v_{y,n,0},\cdots ,v_{y,n,Q-1} \right]}^{T}} \in \mathbb{R}^{Q\times 1}$ denote the velocity vectors, $\forall n$.
Employing the constant velocity model,
we note that
the positions and velocities of all targets at block $n$ can be directly transferred from the state ${{\bm{{\kappa} }}_{A,n-1}}$ at block $n - 1$.
Mathematically,
the updated process for ${{\bm{{\kappa} }}_{A,n}}$ can be expressed as
\begin{align}\label{state_model_part_1}
  & {{\bm{{\kappa }}}_{A,n}}={{\mathbf{{F}}}_{A }}{{\bm{{\kappa }}}_{A,n-1}}+{{\bm{{\omega }}}_{A, n-1}},
\end{align}
where
${{\mathbf{{F}}}_A}=\left[ \begin{matrix}
   {{\mathbf{I}}_{2Q}} & \Delta t{{\mathbf{I}}_{2Q}}  \\
   {{\mathbf{0}}_{2Q}} & {{\mathbf{I}}_{2Q}}  \\
\end{matrix} \right]\in {{\mathbb{R}}^{4Q\times 4Q}}$ indicates
the state transition matrix,
$\Delta t$ is the sensing block interval,
and
${{\bm{{\omega} }}_{A, n-1 }} \in \mathbb{R}^{4Q \times 1}$
is the
process noise of block $n-1$,
representing the uncertainty
associated with the target motion model, $\forall n$.
It is assume that the process noise follows a Gaussian distribution
${{\bm{{\omega} }}_{A, n-1 }} \sim \mathcal{N}\left(\mathbf{0}_{4Q \times 1}, \mathbf{{R}}_A\right)$,
where ${{\mathbf{{R}}_A}}=\mu \left[ \begin{matrix}
   \frac{1}{3}\Delta {{t}^{3}}{{\mathbf{I}}_{2Q}} & \frac{1}{2}\Delta {{t}^{2}}{{\mathbf{I}}_{2Q}}  \\
   \frac{1}{2}\Delta {{t}^{2}}{{\mathbf{I}}_{2Q}} & \Delta t{{\mathbf{I}}_{2Q}}  \\
\end{matrix} \right]$ as described in \cite{book_tracking1},
and
$\mu$ denotes the power of the process noise.

In addition to
target's position and velocity,
as indicated by ${{\bm{{\kappa }}}_{A,n}}$ in (\ref{state_model_part_1}),
the equivalent RCS of each target should also be
included in the target motion model,
as it varies over time.
Next, we will introduce the temporal correlation of the equivalent RCS.
Let
${\bm{{\sigma} }_{all,n, k}}= \left[\operatorname{Re}\left\{ {{\bm{\sigma }}_{R,n, k}^{T}} \right\}, \operatorname{Im}\left\{ {{\bm{\sigma }_{R,n, k}^{T}}} \right\}\right]^{T} \in \mathbb{R}^{2Q\times 1}$
denote the equivalent RCS vector of the $k$-th BS at block $n$,
where
${{\bm{\sigma }}_{R,n, k}}=\left[ {{\sigma }_{R,n,k,0}},{{\sigma }_{R,n,k,1}},\cdots ,{{\sigma }_{R,n,k,Q-1}} \right]^{T}\in {{\mathbb{C}}^{Q\times 1}}$,
$\forall n, k$.
We then define
$\bm{{\sigma} }_{all,n}= \left[{\bm{{\sigma} }^{T}_{all,n, 0}}, \cdots,{\bm{{\sigma} }^{T}_{all,n, K-1}} \right]^{T} \in \mathbb{R}^{2KQ\times 1}$
as the equivalent RCS vector of all targets observed by all BSs at block $n$, $\forall n$.
The state transition for
$\bm{{\sigma} }_{all,n}$ is assumed to be a first-order Markovian process \cite{6095653}, described by the following equation
\begin{align}\label{state_model_part_2}
\bm{{\sigma} }_{all,n}=\bm{{\sigma} }_{all,n-1}+{{\bm{{\omega} }}_{\sigma, n-1 }},
\end{align}
where
${{\bm{{\omega} }}_{\sigma, n - 1 }}\in \mathbb{R}^{2KQ\times 1}$
is the
process
noise at time $n-1$
with the
covariance matrix
${{\mathbf{{R}}}_{\sigma}}$, $\forall n$.
For a typical ISAC networked system,
the fluctuation of RCS,
which is generally a complicated parameter
and can be reflected by ${{\mathbf{{R}}}_{\sigma}}$,
can be characterized statistically using standard Swerling models \cite{RadarHandbook}.
Since the Swerling models are highly dependent on specific sensing scenarios,
it is preferable to determine the RCS's distribution based on the practical applications.

Now,
we combine
(\ref{state_model_part_1})
and
(\ref{state_model_part_2}) into a compact form.
Let
\begin{align}
\nonumber {{\bm{\kappa }}_n}&={{\left[ {{\bm{{\kappa} }}_{A,n}^{T}},\bm{{\sigma} }_{all,n}^{T} \right]}^{T}}\\
\nonumber & = {\left[ {\mathbf{x}_{x, n}^{T}},{\mathbf{x}_{y, n}^{T}},{\mathbf{v}_{x, n}^{T}},{\mathbf{v}_{y, n}^{T}},\bm{{\sigma} }_{all,n}^{T} \right]}\in {{\mathbb{R}}^{\left(4Q + 2KQ\right)\times 1}}
\end{align}
denote the total state vector at block $n$.
The target motion model can finally be given by
\begin{align}\label{state_model_all}
  & {{\bm{\kappa }}_n}={{\mathbf{F}}_{\kappa }}{{\bm{\kappa }_{n-1}}}+\bm{\omega }_{\kappa, n - 1 },
\end{align}
where
${{\mathbf{F}}_{\kappa }}=\text{blkdiag}\left\{ {{\mathbf{{F}}}_{A }},{{\mathbf{I}}_{2{{K}^{2}}Q}} \right\}$
is the total state transition matrix,
and
$\bm{\omega }_{\kappa,n -1  }={{\left[ {{\bm{{\omega} }}_{A, n-1 }^{T}},{{\bm{{\omega} }}_{\sigma, n-1 }^{T}} \right]}^{T}}$
denotes the total process noise vector at sensing block $n-1$
with the covariance of
${{\mathbf{R}}_{\kappa }}=\text{blkdiag}\left\{ {\mathbf{{R}}_A},{{\mathbf{{R}}}_{\sigma}} \right\}$.

\subsection{Networked Measurement Model for Characterizing Geometrical Correlation}\label{Section_networked_measurement_model}

Our networked measurement model
will
establish the connection
between
the target motion state at block $n$, i.e.,
${{\bm{\kappa }}_n}$ as given in (\ref{state_model_all}),
and
the estimated parameters among all BSs at block $n$,
i.e.,
${\hat{\theta }_{n,k,q}}$,
${\hat{\tau }_{n,k ,q}}$,
${\hat{f}_{n ,k,q}}$,
and
${\hat{\xi}_{n,k,q}}$.

Specifically,
we first collect all the target parameters to be estimated into one vector.
As mentioned in Section \ref{Section_dlkjflskdjflksjdf},
we will adopt
${\hat{\theta }_{n,k,q}}$,
${\hat{\tau }_{n,k ,q}}$,
${\hat{f}_{n ,k,q}}$,
and
${\hat{\xi}_{n,k,q}}$
for data fusion.
Let
${{\bm{\theta }}_{n, k}}=\left[ {{\theta }_{n,k,0}},{{\theta }_{n,k,1}},\cdots , {{\theta }_{n,k,Q-1}}\right]^{T}\in {{\mathbb{R}}^{Q\times 1}}$,
${{\bm{\tau }}_{n, k}}=\left[ {{\tau }_{n,k,0}},{{\tau }_{n,k,1}},\cdots , {{\tau }_{n,k,Q-1}}\right]^{T} \in \mathbb{R}^{Q \times 1}$,
${{\mathbf{f}}_{n, k}}=\left[ f_{n,k,0}, f_{n,k,1}, \cdots ,f_{n,k,Q-1}\right]^{T}\in {{\mathbb{R}}^{Q\times 1}}$,
and
${{\bm{\xi }}_{n, k}}=\left[ {{\xi }_{n,k,0}},{{\xi }_{n,k,1}},\cdots , {{\xi }_{n,k,Q-1}}\right]^{T} \in \mathbb{C}^{Q \times 1}$
denote the vectors
containing AOAs, delays, Doppler frequencies,
and complex coefficients
of all targets
observed by the $k$-th BS at block $n$,
respectively, $\forall n, k$.
These vectors are then combined into a single vector
${{\mathbf{\tilde{u} }}_{n, k}} = \left[{{\bm{\theta }}^{T}_{n, k}}, {{\bm{\tau }}^{T}_{n, k}}, {{\mathbf{f}}^{T}_{n, k}}, \operatorname{Re}\left\{ {{\bm{\xi }}_{n, k}^{T}} \right\}, \operatorname{Im}\left\{ {{\bm{\xi }_{n, k}^{T}}} \right\}  \right]^T \in \mathbb{R}^{5Q \times 1}$
to represent the target parameter vector from the received sensing signals of the $k$-th BS
at block $n$,
and
we now define the total parameter vector at block $n$ by
\begin{align}\label{slfjksdlkfj}
{{\mathbf{u }}_{n}}={{\left[ {{\mathbf{\tilde{u} }}^{T}_{n, 0}}, {{\mathbf{\tilde{u} }}^{T}_{n, 1}}, \cdots, {{\mathbf{\tilde{u} }}^{T}_{n, K-1}} \right]}^{T}}\in {{\mathbb{R}}^{5KQ\times 1}}.
\end{align}

Then,
the measurement model
at sensing block $n$, which is nonlinear, can be expressed
as\footnote{Note that the nonlinear transform operator $h\left(\cdot\right)$
is a mapping function for each target.
Therefore, target association is crucial to ensure the correct ordering of target parameters in the vector $\mathbf{u}_n$.
We assume that target association is successfully performed, as this topic has been extensively studied in numerous works, including those focusing on ISAC \cite{9724258} and radar systems \cite{6225454, 6731596}.}
\begin{align}\label{measurement_model}
{{\mathbf{\hat{u}}_n}}=h\left( {{\bm{\kappa }_n}} \right)+{{\bm{\omega }}_{u,n}},
\end{align}
where
${{\mathbf{\hat{u}}}_n} \in {{\mathbb{R}}^{\left( KQ+4{{K}}Q \right)\times 1}}$
is the
measurement vector at sensing block $n$,
and
${{\bm{\omega }}_{u,n}}\in {{\mathbb{R}}^{5KQ\times 1}}$
denotes the measurement noise at block $n$
obeying
${{\bm{\omega }}_{u,n}} \sim \mathcal{N}\left(\mathbf{0}_{5KQ\times 1}, \mathbf{R}_{u,n }\right)$.
Moreover,
$h\left( \cdot  \right):{{\mathbb{R}}^{\left(4Q + 2KQ\right)\times 1}}\mapsto {{\mathbb{R}}^{5KQ\times 1}}$
indicates a nonlinear transform operator with
$h\left( {{\bm{\kappa }_n}} \right) = {{\mathbf{u }_n}}$,
which
bridges
the state vector $\bm{\kappa}_n$ and
the measurement parameters $\mathbf{u}_n$.
This mapping can be explicitly defined by
(\ref{parameter_transform}) and (\ref{RCS_fomula}).

Note that
the covariance matrix
$\mathbf{R}_{u,n }$ of the noise in (\ref{measurement_model})
is
tightly associated with the value of ${{\mathbf{u }_n}}$.
Hence,
$\mathbf{R}_{u,n }$
is related to $n$
and
is
a time-varying matrix
due to different ${{\mathbf{u }_n}}$ in different sensing blocks.
In this paper,
we consider an ideal estimation case
where
the estimates are unbiased and achieve the lowest estimation covariance.
Under such cases,
the covariance matrix $\mathbf{R}_{u,n }$
can be
defined by the Cram\'{e}r-Rao bound (CRB).
Referring to \cite{FSSPVI},
the Fisher information matrix (FIM),
which is the inverse of the CRB,
is given by
\begin{align}\label{CRB_for_signal_model}
{{\left[ \mathbf{R}_{u,n}^{-1} \right]}_{i,j}}=-\mathbb{E}_{\left. {{\mathbf{y}}_{n}} \right|{{\bm{\psi }}_{n}}}\left\{ {\frac{{{\partial }^{2}}\ln p\left( \left. {{\mathbf{y}}_{n}} \right|{{\bm{\psi }}_{n}} \right)}{\partial {{\left[ {{\mathbf{u }_n}} \right]}_{i}}\partial {{\left[ {{\mathbf{u }_n}} \right]}_{j}}}} \right\}.
\end{align}
where
${{\mathbf{y}}_{n}} = \left[{{\mathbf{y}}_{n,{{0}},0,0}}, {{\mathbf{y}}_{n,{{0}},0,1}},\cdots, {{\mathbf{y}}_{n,{{0}},I-1,M-1}},{{\mathbf{y}}_{n,{{1}},0,0}}, \cdots,\right.$
$\left.{{\mathbf{y}}_{n,{{1}},I-1,M-1}}, \cdots,
{{\mathbf{y}}_{n,{{K-1}},I-1,M-1}} \right]^T$
collects the
echo signals
in
frequency-domain
at sensing block $n$.

It is worth noting that
the operator $h\left( \cdot  \right)$
in (\ref{measurement_model})
maps each target's position and velocity to the corresponding target parameters at each BS.
By doing so, $h\left( \cdot  \right)$
effectively characterizes multi-view observations
of the same target,
thereby
our measurement model
unveils
the geometrical correlation in networked sensing.

\subsection{NKF Algorithm}\label{section_networked_filtering_process}
So far,
we have revealed the
temporal
and
geometrical
correlations
by the networked motion model and networked measurement model,
respectively.
In this subsection,
we will
leverage the correlation models to
develop the NKF
algorithm
and
achieve multidimensional data fusion.

The measurement model given in (\ref{measurement_model})
is nonlinear.
To address this problem,
we
propose an NKF algorithm.
Let
${{\bm{\hat{\kappa} }}_{\left. n \right|n}}$
and
${{\bm{\hat{\kappa} }}_{\left. n \right|n-1}}$
represent the posterior and prediction estimates
at block $n$, respectively.
Their corresponding covariance matrices are defined as
${{\mathbf{P}}_{\left. n \right|n-1}}=\mathbb{E}\left\{ \left( {{\bm{{\kappa} }}_{n}}-{{\bm{\hat{\kappa} }}_{\left. n \right|n-1}} \right){{\left( {{\bm{{\kappa} }}_{n}}-{{\bm{\hat{\kappa} }}_{\left. n \right|n-1}} \right)}^{T}} \right\}$
and
${{\mathbf{P}}_{\left. n \right|n}}=\mathbb{E}\left\{ \left( {{\bm{{\kappa} }}_{n}}-{{\bm{\hat{\kappa} }}_{\left. n \right|n}} \right){{\left( {{\bm{{\kappa} }}_{n}}-{{\bm{\hat{\kappa} }}_{\left. n \right|n}} \right)}^{T}} \right\}$.
Similar to the typical Kalman filter,
our proposed NKF algorithm also
contains the prediction step and the update step,
which are detailed next.

First,
for the prediction step,
the state prediction is implemented.
The state prediction at block $n$
can be calculated from the posterior estimate using (\ref{state_model_all})
\begin{align}\label{3o8inrfolwiejhfoiwehf}
{{\bm{\hat{\kappa }}}_{\left. n \right|n-1}}={{\mathbf{F}}_{\kappa }}{{\bm{\hat{\kappa }}}_{\left. n-1 \right|n-1}},
\end{align}
and its covariance matrix is expressed as
\begin{align}
{{\mathbf{P}}_{\left. n \right|n-1}}={{\mathbf{F}}_{\kappa }}{{\mathbf{P}}_{\left. n-1 \right|n-1}}\mathbf{F}_{\kappa }^{T}+{{\mathbf{R}}_{\kappa }}.
\end{align}

Then,
for the updated step,
the posterior estimate
is refined by the measurement vector $\mathbf{u}_n$
\begin{align}\label{sdlifsl}
  {{{\bm{\hat{\kappa }}}}_{\left. n \right|n}}={{{\bm{\hat{\kappa }}}}_{\left. n \right|n-1}}+{{\mathbf{G}}_{n}}\left( {{\mathbf{u}}_{n}}-h\left({{{\bm{\hat{\kappa }}}}_{\left. n \right|n-1}}\right) \right),
\end{align}
where
${{\mathbf{G}}_{n}}\in {{\mathbb{R}}^{ \left( 4Q+2{{K}}Q \right)\times 5KQ}}$
is the Kalman gain, given by
\begin{align}\label{sdliwwwwwwfsl}
{{\mathbf{G}}_{n}}={{\mathbf{P}}_{\left. n \right|n-1}}\mathbf{H}_{n}^{T}{{\left[ {{\mathbf{H}}_{n}}{{\mathbf{P}}_{\left. n \right|n-1}}\mathbf{H}_{n}^{T}+{{\mathbf{R}}_{u,n}} \right]}^{-1}}.
\end{align}
Here,
${{\mathbf{H}}_{n}} \in {{\mathbb{R}}^{5KQ\times
\left( 4Q+2{{K}}Q \right)}}$
is the Jacobian matrix
at ${{{\bm{\kappa }}_{n}}={{{\bm{\hat{\kappa }}}}_{\left. n \right|n-1}}}$,
expressed as
${{\mathbf{H}}_{n}}={{\left. \frac{\partial {{h}}\left( {{\bm{\kappa }}_{n}} \right)}{\partial {{\bm{\kappa }}_{n}^{T}}} \right|}_{{{\bm{\kappa }}_{n}}={{{\bm{\hat{\kappa }}}}_{\left. n \right|n-1}}}}$.
It is worth noting that
the complex coefficient
$\xi_{n,{{k}},q}$
is related to the target position,
as the variable
$\xi_{n,{{k}},q}$
encounters the propagation attenuation, as shown in (\ref{RCS_fomula}).
Finally,
the updated posterior covariance matrix is given by
\begin{align}\label{3o8inrfolwiejhfoiwehf1515}
{{\mathbf{P}}_{\left. n \right|n}}=\left( \mathbf{I}_{4Q+2{{K}}Q}-{{\mathbf{G}}_{n}}{{\mathbf{H}}_{n}} \right){{\mathbf{P}}_{\left. n \right|n-1}}.
\end{align}

With given beamformers,
equations from (\ref{3o8inrfolwiejhfoiwehf}) to (\ref{3o8inrfolwiejhfoiwehf1515}) form a
recursive method for solving the nonlinear tracking problem
and achieve the data fusion.
Algorithm \ref{Algorithm1} summarizes
the main steps of the
proposed NKF algorithm.

\begin{algorithm}[!t]
\caption{Algorithm of the proposed NKF algorithm with given beamformer}\label{Algorithm1}
\label{Algorithm1}
\hspace*{0.02in}{\bf Input:}
the given beamformer ${{\mathbf{{B}}}_{n,{{k}}}}$ of the $k$-th BS
at sensing block $n$,
the received signal $\mathbf{y}_{n,{{k}},i,m}$
at the $m$-th subcarrier of the $i$-th OFDM symbol of the $k$-th BS
at sensing block $n$,
and the covariance matrix of the process noise ${{\mathbf{R}}_{\kappa }}$.\\
\hspace*{0.02in}{\bf Output:}
the posterior estimate
${{{\bm{\hat{\kappa }}}}_{\left. n \right|n}}$,
and
the
state prediction
${{{\bm{\hat{\kappa }}}}_{\left. n + 1 \right|n}}$.

\begin{algorithmic}[1]

\STATE
Estimating the parameters of each target at each BS
and
obtaining the estimates
${\hat{\theta }_{n,k,q}}$,
${\hat{\tau }_{n,k ,q}}$,
${\hat{f}_{n ,k,q}}$,
and
${\hat{\xi}_{n,k,q}}$.
\STATE
Obtaining the measurement vector ${{\mathbf{\hat{u}}_n}}$ based on (\ref{slfjksdlkfj}).\\
\STATE
Calculating $\mathbf{R}_{u,n,k}$ and ${{\mathbf{H}}_{n,k}}$
using ${{\mathbf{\hat{u}}_n}}$.
\STATE
Performing (\ref{3o8inrfolwiejhfoiwehf})-(\ref{3o8inrfolwiejhfoiwehf1515})
to apply the proposed NKF algorithm
and getting the posterior estimate
${{\bm{\hat{\kappa }}}_{\left. n \right|n}}$
and the state prediction
${{\bm{\hat{\kappa }}}_{\left. n+1 \right|n}}$.

\end{algorithmic}
\end{algorithm}

\section{Performance Bounds and Analysis for the Networked Tracking Scheme}\label{Section_performance_PCRB_analysis}

In this section, we present the performance bounds
for the proposed NKF algorithm,
along
with a detailed analysis.

\subsection{Performance Bounds of the Networked Tracking Scheme}

The PCRB serves as a lower bound of the MSE of
any unbiased estimator for target tracking.
This bound, as typically defined by Bayesian information matrix (BIM) $\mathbf{J}_{B}\left( {{\bm{\kappa }}_{n}} \right)$,
follows the Bayesian Cram\'{e}r-Rao inequality
\begin{align}
\mathbb{E}\left\{ \left( {{{\bm{\hat{\kappa }}}}_{\left. n \right|n}}-{{\bm{\kappa }}_{n}} \right){{\left( {{{\bm{\hat{\kappa }}}}_{\left. n \right|n}}-{{\bm{\kappa }}_{n}} \right)}^{T}} \right\}\succeq \mathbf{J}_{B}^{-1}\left( {{\bm{\kappa }}_{n}} \right),
\end{align}
and the BIM consists of two additive parts
\begin{align}\label{dslifjsf}
\mathbf{J}_{B}\left( {{\bm{\kappa }}_{n}} \right) = \mathbf{J}_{P}\left( {{\bm{\kappa }}_{n}} \right) + \mathbf{J}_{D}\left( {{\bm{\kappa }}_{n}} \right),
\end{align}
where
$\mathbf{J}_{P}\left( {{\bm{\kappa }}_{n}} \right)$
and
$\mathbf{J}_{D}\left( {{\bm{\kappa }}_{n}} \right)$
represent
the
prior information
and
the information obtained from the data,
respectively.
Motivated by \cite{668800}, we propose an elegant method for recursively calculating the BIM
by
leveraging the BIM from the previous time step.
The prior information $\mathbf{J}_{P}\left( {{\bm{\kappa }}_{n}} \right)$
can be computed by
\begin{align}\label{4980h98hgouih}
\nonumber {{\mathbf{J}}_{P}}\left( {{\bm{\kappa }}_{n}} \right) & =-{{\mathbb{E}}_{{{\bm{\kappa }}_{n}}}}\left\{ \nabla _{{{\bm{\kappa }}_{n}}}^{{{\bm{\kappa }}_{n}}}\ln p\left( {{\bm{\kappa }}_{n}} \right) \right\} \\
\nonumber & =\mathbf{R}_{\kappa }^{-1}-\mathbf{R}_{\kappa }^{-1}{{\mathbf{F}}_{\kappa }}{{\left( {{\mathbf{J}}_{B}}\left( {{\bm{\kappa }}_{n-1}} \right)+\mathbf{F}_{\kappa }^{T}\mathbf{R}_{\kappa }^{-1}{{\mathbf{F}}_{\kappa }} \right)}^{-1}}\mathbf{F}_{\kappa }^{T}\mathbf{R}_{\kappa }^{-1} \\
 & ={{\left( {{\mathbf{R}}_{\kappa }}+{{\mathbf{F}}_{\kappa }}\mathbf{J}_{B}^{-1}\left( {{\bm{\kappa }}_{n-1}} \right)\mathbf{F}_{\kappa }^{T} \right)}^{-1}}.
\end{align}
Then, the matrix $\mathbf{J}_{D}\left( {{\bm{\kappa }}_{n}} \right)$
can be expressed as
\begin{align}\label{114980h98hgouih}
\nonumber {{\mathbf{J}}_{D}}\left( {{\bm{\kappa }}_{n}} \right) &=-{{\mathbb{E}}_{{{\bm{\kappa }}_{n}},{{\mathbf{Y}}_{n}}}}\left\{ \nabla _{{{\bm{\kappa }}_{n}}}^{{{\bm{\kappa }}_{n}}}\ln p\left( {{\mathbf{Y}}_{n}}\left| {{\bm{\kappa }}_{n}} \right. \right) \right\} \\
\nonumber  & ={{\mathbb{E}}_{{{\bm{\kappa }}_{n}}}}\left\{ {{\mathbb{E}}_{{{\mathbf{Y}}_{n}}\left| {{\bm{\kappa }}_{n}} \right.}}\left\{ -\nabla _{{{\bm{\kappa }}_{n}}}^{{{\bm{\kappa }}_{n}}}\ln p\left( {{\mathbf{Y}}_{n}}\left| {{\bm{\kappa }}_{n}} \right. \right) \right\} \right\} \\
\nonumber & ={{\mathbb{E}}_{{{\bm{\kappa }}_{n}}}}\left\{ \mathbf{H}_{n}^{T}\mathbf{R}_{u,n}^{-1}{{\mathbf{H}}_{n}} \right\},\\
& =\mathbf{H}_{n}^{T}\mathbf{R}_{u,n}^{-1}{{\mathbf{H}}_{n}},
\end{align}
where
${{\mathbf{H}}_{n}}$ is defined in (\ref{sdliwwwwwwfsl}),
and
$\mathbf{H}_{n}$ and $\mathbf{R}_{u,n}$
are independent of ${{\bm{\kappa }}_{n}}$ and can be determined by the prediction estimate
from the previous time.

Combining both matrices in (\ref{4980h98hgouih}) and (\ref{114980h98hgouih}),
we can obtain the BIM
\begin{align}\label{dfodojfdojfodjf333dfdf}
  \nonumber {{\mathbf{J}}_{B}}\left( {{\bm{\kappa }}_{n}} \right)=& \ {{\left( {{\mathbf{R}}_{\kappa }}+{{\mathbf{F}}_{\kappa }}\mathbf{J}_{B}^{-1}\left( {{\bm{\kappa }}_{n-1}} \right)\mathbf{F}_{\kappa }^{T} \right)}^{-1}}\\
 &+\mathbf{H}_{n}^{T}\mathbf{R}_{u,n}^{-1}{{\mathbf{H}}_{n}},
\end{align}
where the initial matrix of the BIM
can be calculated by
${{\mathbf{J}}_{B}}\left( {{\bm{\kappa }}_{0}} \right)=\mathbf{H}_{n}^{T}\mathbf{R}_{u,n}^{-1}{{\mathbf{H}}_{n}}$.

For networked localization systems
without applying target tracking algorithms \cite{9724258},
the CRB is only a function of ${{\mathbf{J}}_{D}}\left( {{\bm{\kappa }}_{n}} \right)$.
Therefore,
it can be readily concluded that
our networked tracking scheme
outperforms the conventional non-tracking scheme,
because both ${{\mathbf{J}}_{P}}\left( {{\bm{\kappa }}_{n}} \right)$
and
${{\mathbf{J}}_{D}}\left( {{\bm{\kappa }}_{n}} \right)$
are the positive definite matrices
in general sensing scenarios.
The superior precision and reliability of our tracking-based scheme
will be validated
through simulation results in Section \ref{Section_simulation_results}.
Next, we examine further properties using the derived PCRB of our scheme.

\subsection{Analysis for the PCRB of the Networked Tracking Scheme}\label{Section_analysis_PCRB}
In this subsection, some properties of the PCRB in (\ref{dfodojfdojfodjf333dfdf}) are unveiled,
emphasizing the impact of different factors on the networked tracking performance.

We begin with presenting the following proposition to assess the feasibility of networked tracking
with respect to the locations of targets and BSs.
\begin{proposition}\label{analysis_feasibility_networked_sensing}
In a dual-BS networked sensing scenario,
the matrix $\mathbf{J}_{D}\left( {{\bm{\kappa }}_{n}} \right)$ is singular
when targets are collinear with respect to two BSs.
\end{proposition}
\begin{proof}
It can be readily proved that the matrix
$\mathbf{H}_{n}^{T}{{\mathbf{H}}_{n}}$ becomes singular in this scenario.
Hence,
${{\mathbf{J}}_{D}}\left( {{\bm{\kappa }}_{n}} \right)=\mathbf{H}_{n}^{T}\mathbf{R}_{u,n}^{-1}{{\mathbf{H}}_{n}}$
is also singular.
Proposition \ref{analysis_feasibility_networked_sensing} is thus proved.
\end{proof}

Essentially,
a single BS can only estimate the radial velocity of a target but lacks the capability to determine its tangential velocity.
When the targets' positions satisfy the condition outlined in Proposition \ref{analysis_feasibility_networked_sensing},
the networked velocity sensing is degenerated into the single-BS velocity sensing.
In such cases, the system fails to estimate the tangential velocity of each target.
Hence,
multiple BSs with diverse observation perspectives
are necessary to jointly infer the velocity vector.

Proposition \ref{analysis_feasibility_networked_sensing}
also reveals that
networked tracking cannot work
when the target moves along the line connecting two BSs.
Therefore, an interesting result is provided below.
\begin{corollary}
Two BSs
are generally sufficient, and three non-collinear BSs
can guarantee an effective networked tracking.
\end{corollary}
Assuming that targets are not collinear with two BSs,
a critical proposition for the transmit beamformer is revealed
as follows.
\begin{proposition}\label{proposition_1_context}
The BIM is a linear function with respect to
$\mathbf{R}_{B, n,k,m}$,
where
the covariance matrix $\mathbf{R}_{B, n,k,m}$
is defined in (\ref{formula_signal_transmitting}).
\end{proposition}
\begin{proof}
Please refer to Appendix \ref{proposition1_proof}.
\end{proof}

Due to
$\mathbf{R}_{B, n,k,m}={{{\mathbf{{B}}}}_{n,k,m}}\mathbf{{B}}_{n,k,m}^{H}$,
where
${{{\mathbf{{B}}}}_{n,k,m}}$
contains all the beamformers
as illustrated in (\ref{formula_signal_transmitting}),
the beamformers will significantly affect the tracking
accuracy.
This result suggests that
it is possible to optimize the beamformers to further
improve the networked tracking accuracy.
Next,
we will use this essential proposition to optimize the beamformers in a general multi-target networked sensing scenario.

\section{Beamforming Optimization for Networked Tracking}\label{Section_tracking_optimizatio}

Building on the relationship between beamformers and networked tracking accuracy,
in this section,
we propose a PCRB-based optimization framework
to dynamically optimize the beamforming vectors of each BS
and achieve the target association between targets and BSs.

As depicted in Fig. \ref{Fig_process},
after the networked tracking at the sensing block $n$,
we now proceed to optimize the transmit beamforming vectors
for the subsequential sensing block $n + 1$.
Here,
we adopt the PCRB, i.e., $\mathbf{J}_{B}^{-1}\left( {{\bm{{\kappa }}}_{{ n + 1 }}} \right)$, as an optimization criterion
to devise the beamformers,
subject to the constrains of the transmit power.
Note that the diagonal elements of
$\mathbf{J}_{B}^{-1}\left( {{\bm{{\kappa }}}_{{ n + 1 }}} \right)$
denote the lower bounds of various parameters' estimation variance.
The position and velocity are the primary focus over the tracking process.
Hence,
we introduce the weighting vector $\mathbf{w }={{\left[ {{w }_{0}},\cdots ,{{w }_{4Q-1}},0,\cdots ,0 \right]}^{T}}\in \mathbb{R}^{ \left(4Q + 2KQ\right) \times 1}$
to balance the units associated with different target parameters.
Our \emph{Trace-Opt-based optimization} problem is formulated with respect to matrix traces as
\begin{align}
 \nonumber  \left(\text{P1}\right): \  \underset{{{\mathbf{R}}_{B,n+1,{{k}},m}}\in \mathbb{C}^{N_T \times N_T}}{\mathop{\min }}\,\ &\text{tr}\left\{ \text{diag}\left\{ \mathbf{w } \right\}\mathbf{J}_{B}^{-1}\left( {{\bm{{\kappa }}}_{{ n + 1 }}} \right) \right\} \\
 \text{s.t.}\ &   \sum_m \text{tr}\left\{{{\mathbf{R}}_{B,n+1,{{k}},m}}\right\} \le {{P}_{k}}, \ \forall k,
\end{align}
where $P_{k}$ is the transmit power constraint for the $k$-th BS,
and $\mathbf{R}_{B, n+1,k,m}$
is the covariance matrix of the signals
emitted by the transmit antenna array of BS $k$ of sensing block $n+1$
at the $m$-th subcarrier,
as defined in (\ref{formula_signal_transmitting})

However,
Problem (P1) encounters an essential sequential constraint.
As depicted in Fig. \ref{Fig_process},
beamformers at sensing block $n + 1$ should be optimized during the interval between the sensing blocks $n$ and $n + 1$,
where the measurement vector, i.e., $\mathbf{\hat{u}}_{n + 1}$,
and the state vector, i.e., $\bm{\kappa}_{n+1}$,
for block $n+1$ remain unavailable.
Since the
beamforming
optimization for each sensing block inevitably requires
knowledge of the target parameters of the current sensing block,
this sequential dependency creates a significant challenge for networked beamforming optimization.
Fortunately, the tracking algorithm proposed in Section \ref{section_networked_filtering_process} offers an effective solution.
By leveraging the target motion model, the
tracking step in (\ref{3o8inrfolwiejhfoiwehf})
of the
tracking algorithm
provides predicted target parameters for the subsequent sensing block,
denoted as ${{\bm{\hat{\kappa }}}_{\left. n + 1 \right|n}}$.
Treating ${{\bm{\hat{\kappa }}}_{\left. n + 1 \right|n}}$
as the true target state,
our optimization algorithm is designed to predict
the optimal beamformer for the $\left(n+1\right)$-th sensing block during the $n$-th sensing block.
This predictive approach ensures that the beampattern is highly likely to illuminate the potential target positions in the next block step,
thereby improving both localization accuracy and tracking robustness.
Now, Problem (P1) can be rewritten as
\begin{align}
 \nonumber  \left(\text{P2}\right): \  \underset{{{\mathbf{R}}^{\left({\left. n + 1 \right|n}\right)}_{B,{{k}},m}}\in \mathbb{C}^{N_T \times N_T}}{\mathop{\min }}\,\ &\text{tr}\left\{ \text{diag}\left\{ \mathbf{w } \right\}\mathbf{J}_{B}^{-1}\left( {{\bm{\hat{\kappa }}}_{\left. n + 1 \right|n}} \right) \right\} \\
 \text{s.t.} \ &  \sum_m \text{tr}\left\{{{\mathbf{R}}^{\left({\left. n + 1 \right|n}\right)}_{B,{{k}},m}}\right\} \le {{P}_{k}}, \ \forall k,
\end{align}
where
${{\mathbf{R}}^{\left({\left. n + 1 \right|n}\right)}_{B, k,m}}$
denotes the
prediction of
the $\left(n+1\right)$-th sensing block's
transmit covariance matrix of BS $k$
at the $n$-th sensing block.

For Problem (P2),
the constraint is convex, whereas
the objective function is expressed in a non-convex form.
To address this problem,
we employ
the Schur complement technique \cite{Vandenberghe1998DeterminantMW}
to cast Problem (P2) to a SDP problem, as
\begin{align}
\nonumber  \left( \text{P3} \right): \ \underset{\begin{smallmatrix}
 {{\alpha}_{n+1, a}} , \
  {{\mathbf{R}}^{\left({\left. n + 1 \right|n}\right)}_{B,{{k}},m}}
\end{smallmatrix}}{\mathop{\min }}\,\ & \sum\limits_{a=0}^{4Q-1}{w_a{{\alpha}_{n+1, a}}} \\
\nonumber \text{s.t.}\  \ \ \ \  \ \ \ \  \ & \left[ \begin{matrix}
   {{\mathbf{J}}_{B}}\left( {{\bm{\hat{\kappa }}}_{\left. n + 1 \right|n}} \right) & {{\mathbf{e}}_{a}}  \\
   \mathbf{e}_{a}^{T} & {{\alpha}_{n+1, a}}  \\
\end{matrix} \right]\succeq 0 , \ \forall a\\
\nonumber  &  \sum_m \text{tr}\left\{ {{\mathbf{R}}^{\left({\left. n + 1 \right|n}\right)}_{B, {{k}},m}} \right\}\le {{P}_{k}} , \ \forall {{k}} \\
 & {{\mathbf{R}}^{\left({\left. n + 1 \right|n}\right)}_{B, k,m}}\succeq 0 , \ \forall {{k,m}},
\end{align}
where
$\alpha_{n, a}$ is the auxiliary variable,
and
$\mathbf{e}_a \in \mathbb{Z}^{\left(4Q + 2KQ\right)\times 1}$ denotes the $a$-th column of the identity matrix.
Problem (P3) is a convex
optimization problem and can be efficiently resolved using
CVX.

The optimal solution to Problem (P3),
as well as (P2),
possess an important property as described in Proposition \ref{proposition_optimization_span}.
\begin{proposition}\label{proposition_optimization_span}
Let $\mathbf{{R}}^{\left({\left. n + 1 \right|n}\right)}_{\text{op},B, k,m}$ represent the optimal matrix as the solution to Problem (P3).
The matrix $\mathbf{{R}}^{\text{op}}_{B,n+1, k,m}$ is spanned by
\begin{align}
\nonumber  \mathbf{{R}}^{\left({\left. n + 1 \right|n}\right)}_{\text{op},B, k,m}\in \text{span}& \left\{ {{{\mathbf{\dot{a}}}}_{t}}\left( {\hat{\varphi }^{\left({\left. n + 1 \right|n}\right)}_{k,0}} \right),\cdots ,{{{\mathbf{\dot{a}}}}_{t}}\left( {\hat{\varphi }^{\left({\left. n + 1 \right|n}\right)}_{k,Q-1}} \right), \right. \\
& \left. {{\mathbf{a}}_{t}}\left( {\hat{\varphi }^{\left({\left. n + 1 \right|n}\right)}_{k,0}} \right),\cdots ,{{\mathbf{a}}_{t}}\left( {\hat{\varphi }^{\left({\left. n + 1 \right|n}\right)}_{k,Q-1}} \right) \right\} , \ \forall {{k}},
\end{align}
where
$\hat{\varphi}^{\left({\left. n + 1 \right|n}\right)}_{k, q}$
denotes the
prediction of
the $\left(n+1\right)$-th sensing block's
AOD of target $q$ of BS $k$
at the $n$-th sensing block.
Referring to (\ref{parameter_transform}),
the AOD $\hat{\varphi}^{\left({\left. n + 1 \right|n}\right)}_{k, q}$
can be determined by AOA $\hat{\theta}^{\left({\left. n + 1 \right|n}\right)}_{k, q}$,
which can be obtained from ${{\bm{\hat{\kappa }}}_{\left. n + 1 \right|n}}$.

\end{proposition}
\begin{proof}
Please refer to Appendix \ref{Appendix_proposition_optimization_span}.
\end{proof}

It can be observed that
both the steering vector,
$\mathbf{a}_t\left(\hat{\varphi}^{\left({\left. n + 1 \right|n}\right)}_{k, q}\right)$,
and its derivative vector,
$\mathbf{\dot{a}}_t\left(\hat{\varphi}^{\left({\left. n + 1 \right|n}\right)}_{k, q}\right)$,
are coexist in the optimal solution.
Intuitively,
the vector
$\mathbf{a}_t\left(\hat{\varphi}^{\left({\left. n + 1 \right|n}\right)}_{k, q}\right)$
in $\mathbf{{R}}^{\left({\left. n + 1 \right|n}\right)}_{\text{op},B, k, m}$
increases the
echo power
because it steers the beam towards the $q$-th target,
while the derivative vector
$\mathbf{\dot{a}}_t\left(\hat{\varphi}^{\left({\left. n + 1 \right|n}\right)}_{k, q}\right)$
characterizes the rate of signal attenuation
as it represents the difference-beam of the transmit array.
Therefore,
the combination of both vectors enhances tracking accuracy.

\section{Numerical Results}\label{Section_simulation_results}

\begin{figure}[!t]
\centering
{\includegraphics[width=3in]{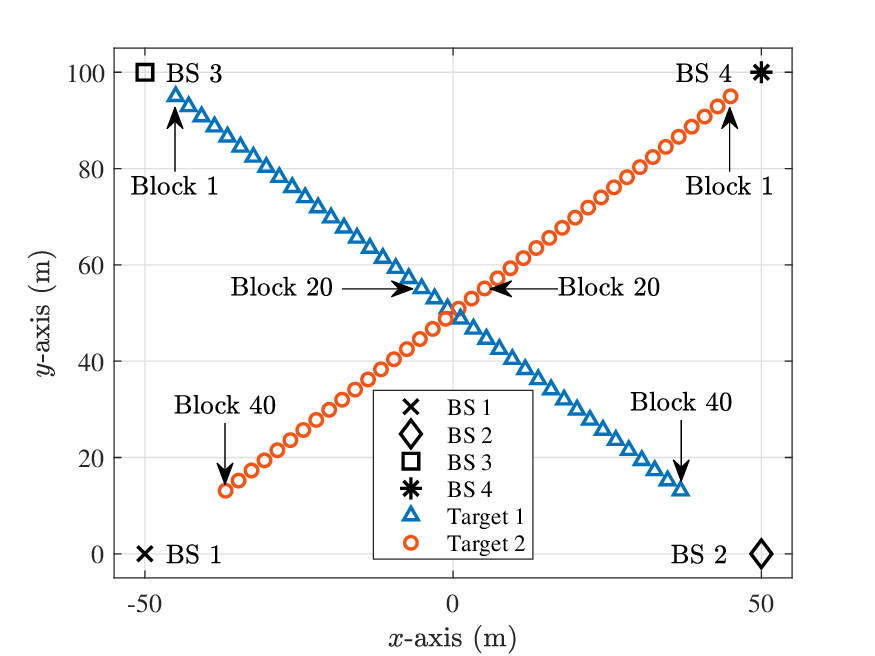}}
\caption{Configuration of the networked sensing scenario with four BSs and two dynamic targets.}
\label{Fig_target_motion_real}
\end{figure}

In this section, simulation results are presented to validate the
proposed networked sensing scheme.
We first describe the system setup,
and
then evaluate the sensing performance.

\begin{table}[]
\caption{Simulation Parameters}\label{table_simulation_parameters_real}
\centering
\begin{tabular}{cc}
\toprule
\toprule
Parameters                                 & Values      \\ \midrule
\makecell[c]{carrier frequency}    & $f_c = 3 \text{GHz}$                  \\
the number of BSs & $K = 4$ \\
position of BSs & $\left(-50\text{m}, 0\text{m}\right)$, $\left(50\text{m}, 0\text{m}\right)$\\
                    & $\left(-50\text{m}, 100\text{m}\right)$, $\left(50\text{m}, 100\text{m}\right)$ \\
the number of targets & $Q = 2$ \\
initial position of targets & $\left(-45\text{m}, 95\text{m}\right)$, $\left(45\text{m}, 95\text{m}\right)$\\

initial velocity of targets & $\left(2.1\text{{m}/{s}}, -2.1\text{{m}/{s}}\right)$,\\
&$\left(-2.1\text{{m}/{s}}, -2.1\text{{m}/{s}}\right)$\\

sensing block interval & $\Delta t = 1\text{s}$\\
the number of transmit antennas&$N_T = 8$\\
in each BS &\\
the number of receive antennas&$N_R = 8$\\
in each BS &\\
antenna interval & $d = {\lambda}/{2}$\\
subcarrier interval  & $\Delta f = 480 \text{kHz}$                  \\
OFDM symbol interval                & $T_0 = 0.1 \text{ms}$                          \\
the number of subcarriers                           & $M = 8$           \\
the number of OFDM symbols                         & $I = 100$                            \\
transmit power of each BS & $40\text{dBmW}$ \\
RCS of each target            & $\sigma_\text{Target} = 1 \text{m}^2$       \\
noise power            & $\sigma_z^2 = 4.92 \times 10^{-12} \text{W}$       \\
\bottomrule
\end{tabular}
\end{table}

The simulation parameters are detailed in Table \ref{table_simulation_parameters_real}.
Based on Table \ref{table_simulation_parameters_real},
Fig. \ref{Fig_target_motion_real}
depicts the BS positions and the target positions
across a total of 40 blocks.
Moreover, the noise power in Table \ref{table_simulation_parameters_real}
is calculated as $\sigma_z^2 = k_B F_n T_{st} B = k_B F_n T_{st} N \Delta f = 1.476 \times 10^{-12} \text{W}$,
where $k_B$ is the
Boltamann constant,
$F_n = 3$ is the receiver noise figure,
and $T_{st} = 290\text{K}$ is the standard temperature.
Additionally,
we allocate ${M}/{K} = 64$
subcarriers to each BS
in an interleaved configuration, following the typical allocation method in 5G NR.
For the $k$-th
BS,
the allocated subcarrier index
is denoted by
$mK + k$,
where $m = 0, 1, \cdots, {M}/{K} - 1$
and $k = 0, 1, \cdots, K - 1$.
Although our proposed scheme can work with any subcarrier configurations,
we adopt the above allocation strategy
due to its capability of maximizing bandwidth for target resolution
and reducing sidelobe for target detection.
Additionally,
the delay, Doppler, AOA,
and AOD for each target at each BS can be estimated using MLE
to ensure a optimal estimation accuracy.

\begin{figure*}[!t]
\centering
{\includegraphics[width=7in]{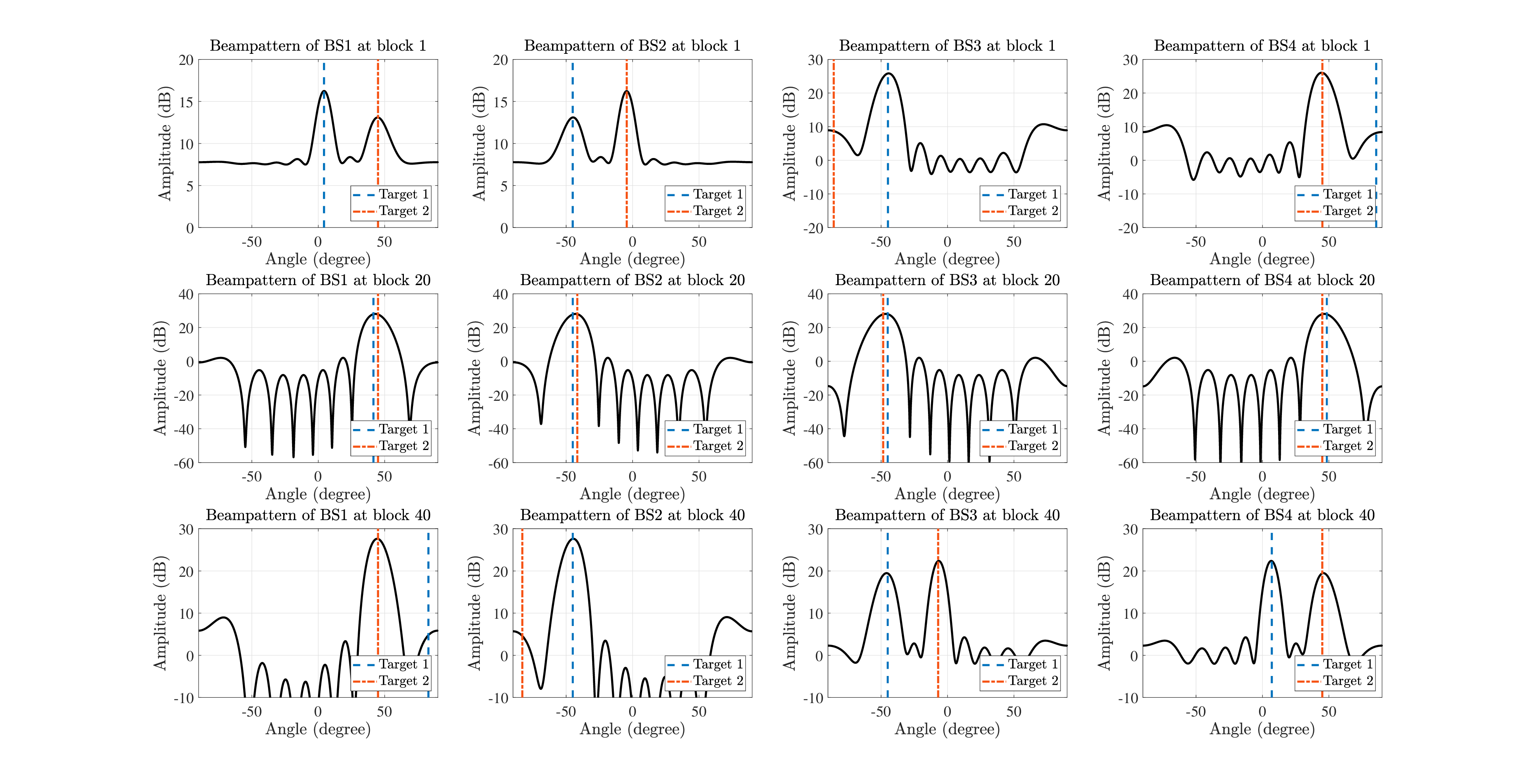}}
\caption{Transmit beampattern results of all BSs at three selected blocks:
block 1, block 20, and block 40.}
\label{Fig_simulation_optimization}
\end{figure*}

Fig. \ref{Fig_simulation_optimization} illustrates
how the transmit beampatterns of all BSs change in response to the targets' motion.
It is worth noting that the
optimal beam direction at each subcarrier is approximately the same,
and thus,
we use the beam at the first subcarrier for illustration.
Three blocks,
i.e., sensing block 1, sensing block 20, and sensing block 40,
are selected from a total of 40 tracking blocks, as marked in Fig. \ref{Fig_target_motion_real}.
The weighting vector for Problem (P3) is set
to $\mathbf{w} = \left[\mathbf{1}^{T}_{2Q \times 1}, \mathbf{0}^{T}_{\left(2Q + 2KQ\right) \times 1}\right]^T$.
At sensing block 1,
BSs 1 and 2 are far away from both targets, BS 3 is close to target 1 and far away from target 2, and BS 4 is close to target 2 and far away from target 1.
Then,
it is observed that
BS 1 and BS 2 generate two lobes,
each directed towards one target,
while BS 3 only transmits one lobe towards target 1,
and BS 4 only transmits one lobe towards target 2.
This indicates that when a BS is close to one target but far away from another target,
it is merely associated with the near target.
When a BS is not close to both targets,
it is associated with both targets.
It is observed that from BS 1 and BS 2 that generates two lobes,
the peak positions of the beampatterns are directed at their nearest target,
which is intuitively expected,
as a closer target experiences lower propagation attenuation.
A similar observation is made at sensing block 40.
At sensing block 20,
both of the two targets are almost of equal distance to the four BSs.
It is observed that
the four BSs adjust their beampatterns to simultaneously track both targets with a single beam,
leading to both targets being illuminated with almost the same power.
Fig. \ref{Fig_simulation_optimization} shows
that the BSs dynamically adjust their transmit signals in response to the locations and velocities of all the targets, as determined by the proposed optimization algorithm.
The observation also reveals that, in all sensing blocks, each target is covered by the beams of at least two BSs.
This is reasonable,
as effective velocity estimation requires the collaborative effort of at least two BSs.
This requirement arises from the fact that using the Doppler can only measure the \emph{radial velocity} between a target and a single BS, whereas utilizing two BSs for a single target enables the estimation of the velocity's direction,
which verifies Proposition \ref{analysis_feasibility_networked_sensing}.

\begin{figure}[!t]
\centering
{\includegraphics[width=3in]{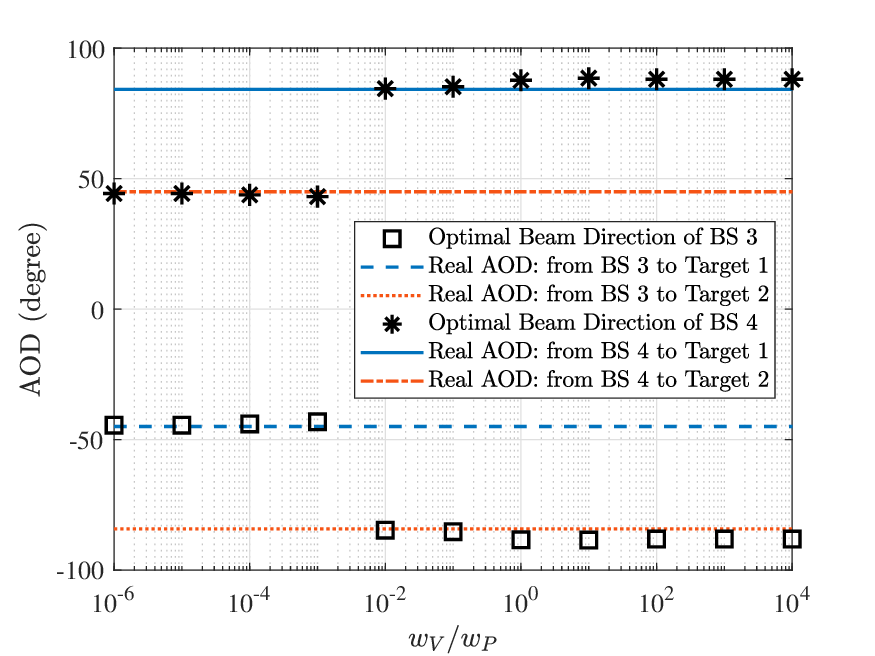}}
\caption{Optimal beam direction versus different weighting vectors for Problem (P3) at the first sensing block.}
\label{Fig_w_velocity_position}
\end{figure}

Fig. \ref{Fig_w_velocity_position}
shows the optimal beam direction of BSs 3 and 4 at the first sensing block under
different weighting vectors for Problem (P3),
where the weighting vector is devised by
$\mathbf{w} = \left[w_P \mathbf{1}^{T}_{2Q \times 1}, w_V \mathbf{1}^{T}_{2Q \times 1}, \mathbf{0}^{T}_{2KQ \times 1}\right]^T$.
As ${w_V}/{w_P} \to 0$,
the optimization result prioritizes tracking the target's position,
whereas ${w_V}/{w_P} \to \infty$,
it focuses on tracking the target's velocity.
It is worth noting that
the beam directions of BSs 1 and 2
consistently
point toward Target 1 and Target 2, respectively,
for different values of ${w_V}/{w_P}$.
For the beams of BSs 3 and 4,
the simulation results in
Fig. \ref{Fig_w_velocity_position}
indicate that
when ${w_V}/{w_P} \to 0$,
the optimal beam direction aligns with the result in Fig. \ref{Fig_simulation_optimization}.
In contrast, as ${w_V}/{w_P} \to \infty$,
the beam direction of BS 3 gradually adjusts to illuminate Target 2,
while BS 4 reorients toward Target 1.
These entirely different results
revealed an essentially distinction of the velocity measurement from the localization:
increasing the difference of observation angles
contributes to enhance the accuracy of the velocity.
Take Target 1 as an example.
Compared to BSs 1 and 3,
BSs 1 and 4 provides more distinct observation angles,
thus increasing the velocity accuracy.
This suggests a trade-off in beam allocation between velocity and position tracking.
Therefore,
the weighting vector $\mathbf{w}$ needs to be carefully designed based on practical requirements.

\begin{figure}[!t]
\centering
\subfigure[]
{\includegraphics[width=3in]{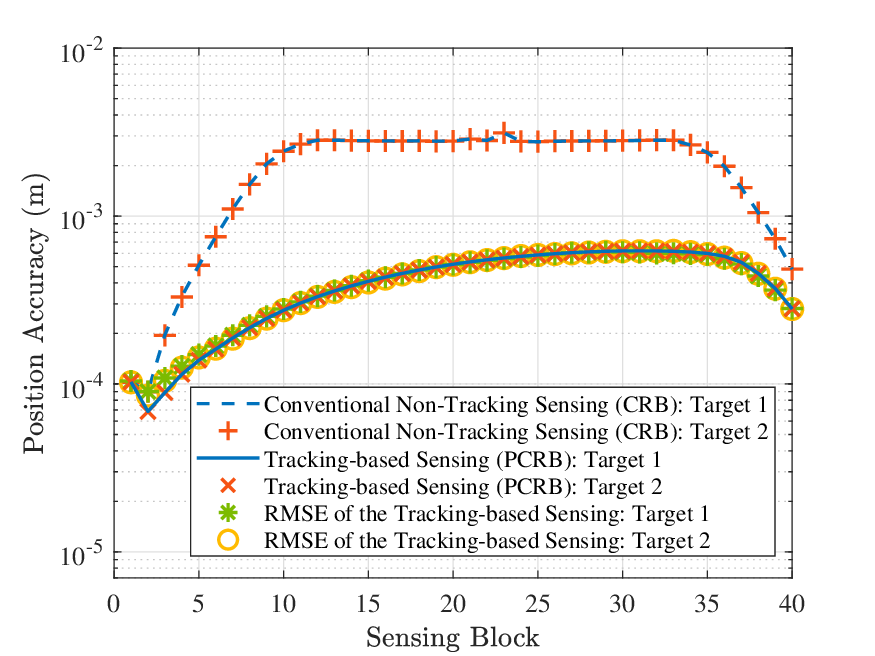}}
\subfigure[]
{\includegraphics[width=3in]{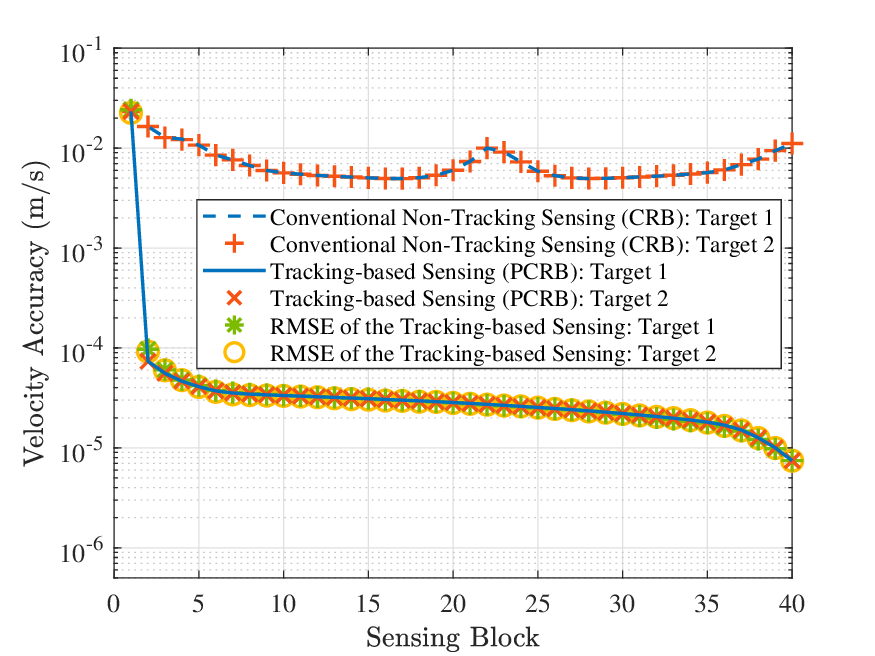}}
\caption{Comparisons of the tracking accuracy under different blocks,
with
(a) position accuracy of the tracking-based sensing versus that of the conventional non-tracking sensing,
and
(b) velocity accuracy of the tracking-based sensing versus that of the conventional non-tracking sensing.}
\label{Fig_target_comparison}
\end{figure}

\begin{figure}[!t]
\centering
\subfigure[]
{\includegraphics[width=3in]{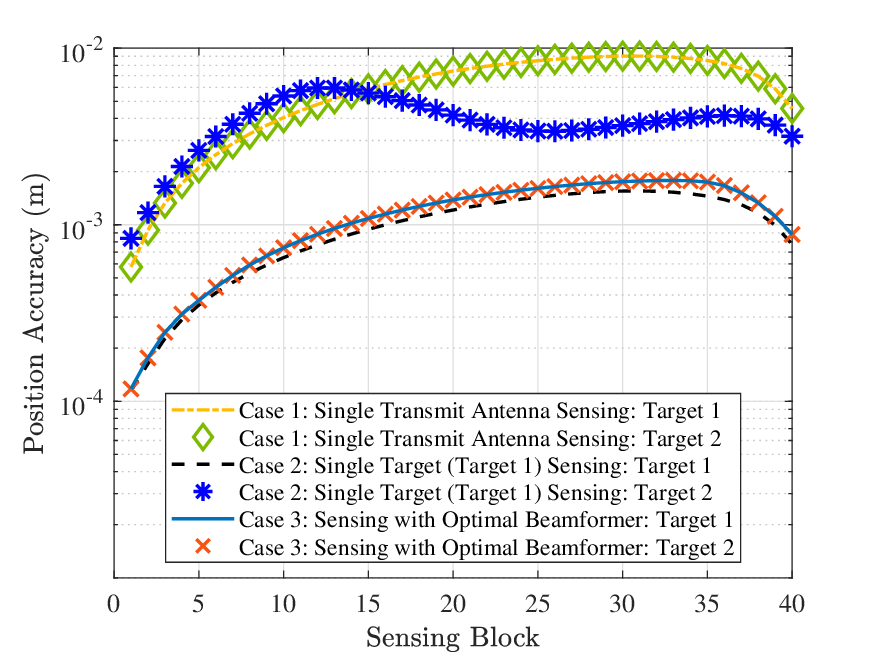}}
\subfigure[]
{\includegraphics[width=3in]{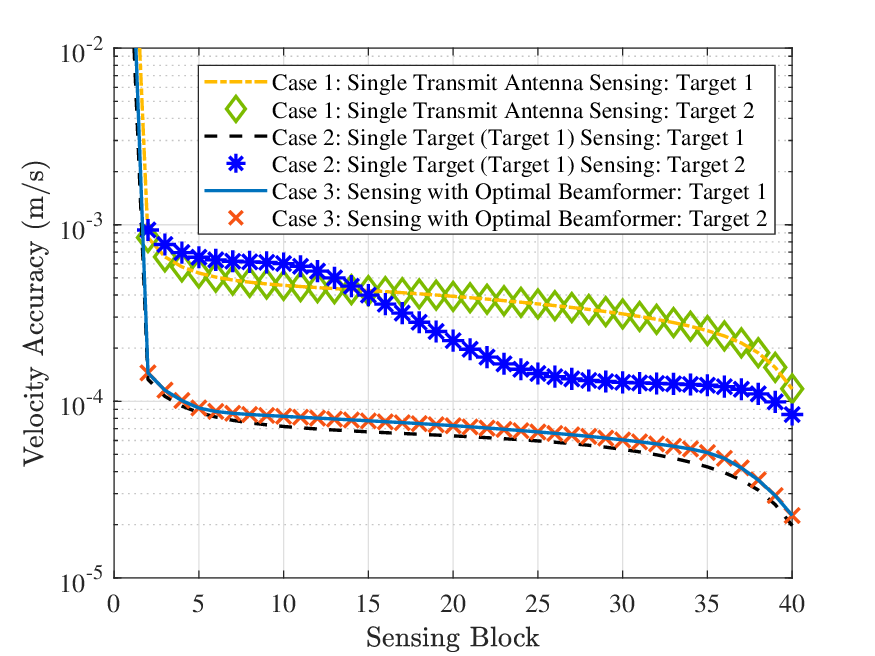}}
\caption{Comparisons of three sensing cases under different blocks,
with
(a) position accuracy,
and
(b) velocity accuracy.}
\label{Fig_target_necessity}
\end{figure}

Fig. \ref{Fig_target_comparison}
shows a comparison between
the conventional non-tracking sensing and tracking-based sensing
in terms of
the position and velocity accuracy.
Note that the weighting vector is chosen as
$\mathbf{w} = \left[\mathbf{1}^{T}_{2Q \times 1}, \mathbf{0}^{T}_{\left(2Q + 2KQ\right) \times 1}\right]^T$,
i.e., focusing on the tracking of targets' position.
The position and velocity estimation accuracy of the tracking-based sensing is defined by the PCRB,
given by
\begin{align}
 \nonumber & {{\delta }_{x, n, q}}=\sqrt{\left[ {{\mathbf{J}}_{B}}\left( {{\bm{\kappa }}_{n}} \right) \right]_{{{x}_{x,n,q}}}^{2}+\left[ {{\mathbf{J}}_{B}}\left( {{\bm{\kappa }}_{n}} \right) \right]_{{{x}_{y,n,q}}}^{2}} \\
 & {{\delta }_{v, n, q}}=\sqrt{\left[ {{\mathbf{J}}_{B}}\left( {{\bm{\kappa }}_{n}} \right) \right]_{{{v}_{x,n,q}}}^{2}+\left[ {{\mathbf{J}}_{B}}\left( {{\bm{\kappa }}_{n}} \right) \right]_{{{v}_{y,n,q}}}^{2}},
\end{align}
where ${{\delta }_{{x, n, q}}}$
and
${{\delta }_{v, n, q}}$
denote the position and velocity accuracy, respectively.
The RMSE of Algorithm \ref{Algorithm1} is shown to closely align with the optimal PCRB in both figures, demonstrating that the proposed algorithm is effective and the optimal PCRB is achievable.
Moreover,
the tracking-based position errors, as depicted in Fig. \ref{Fig_target_comparison} (a),
initially increase and then decrease
as the number of the sensing blocks increases.
This behavior is attributed to the increased
propagation attenuation of sensing signals
when the targets are closer to the center of Fig. \ref{Fig_target_motion_real}.
In contrast, the tracking-based velocity errors,
shown in Fig. \ref{Fig_target_comparison} (b),
gradually decrease as more blocks are included in the tracking process.
This phenomenon arises from the dependence of velocity estimation
on the target motion model, where a constant-velocity model is assumed,
thereby a larger blocks leading to a more accurate velocity estimation.
Furthermore,
it can also be concluded that
both targets exhibit the same sensing accuracy
as their motion is symmetric with respect to $x = 0$.
Additionally,
for both sensing results in Figs. \ref{Fig_target_comparison} (a) and (b),
the tracking-based approach outperforms
the conventional non-tracking sensing.
This demonstrates that the tracking process
can effectively reduce the sensing errors leveraging the observations over multiple blocks with the prior information of the target motion model.

Fig. \ref{Fig_target_necessity}
demonstrates the superiority of the beamforming optimization.
Three sensing cases are provided for comparisons:
\begin{itemize}
\item[(1)]
Case 1:
Networked sensing with a single transmit antenna (single transmit antenna sensing in Fig. \ref{Fig_target_necessity});
\item[(2)]
Case 2:
Networked sensing with only focusing Target 1 (single target sensing in Fig. \ref{Fig_target_necessity}).
\item[(3)]
Case 3:
Networked sensing with the PCRB-based optimal transmit beamformer (sensing with optimal beamformer in Fig. \ref{Fig_target_necessity});
\end{itemize}
For a fair comparison,
we ensure that
the beamformer in
these sensing cases
transmits with the same power,
i.e.,
$\text{tr}\left\{ \mathbf{R}_{B,n,{{k},m}} \right\}={{P}_{{{k,m}}}}$.
In networked sensing Case 1,
omnidirectional sensing signals are transmitted with equal power in all directions.
Due to the power constraint,
the signals transmitted to non-target regions
result in power loss,
thereby reducing the sensing accuracy for both targets.
In networked sensing Case 2,
all the BSs adjust their beamformers to focus on the direction of Target 1.
This beamforming strategy is optimal for a single-target sensing scenario,
where the PCRB for Target 1 achieves the lowest value among all sensing cases.
However,
the sensing accuracy for Target 2 is significantly reduced due to the beamformer mismatch,
so the single-target sensing strategy is not optimal for multi-target scenarios.
In networked sensing Case 3,
although the accuracy for both targets does not reach the accuracy of Target 1 in Case 2,
Case 3 ensures that all targets maintain a low tracking PCRB,
making it a promising approach for multi-target tracking.
Furthermore,
based on the simulation results in Fig. \ref{Fig_target_necessity},
the accuracy of the two targets in Case 3 suffers only a 1.12 times loss compared to the optimal accuracy in single-target sensing,
which is considered an acceptable trade-off.
To sum up,
the results above
demonstrates the superiority of
the proposed algorithm in multi-target tracking scenarios.

\begin{figure}[!t]
\centering
\subfigure[]
{\includegraphics[width=3in]{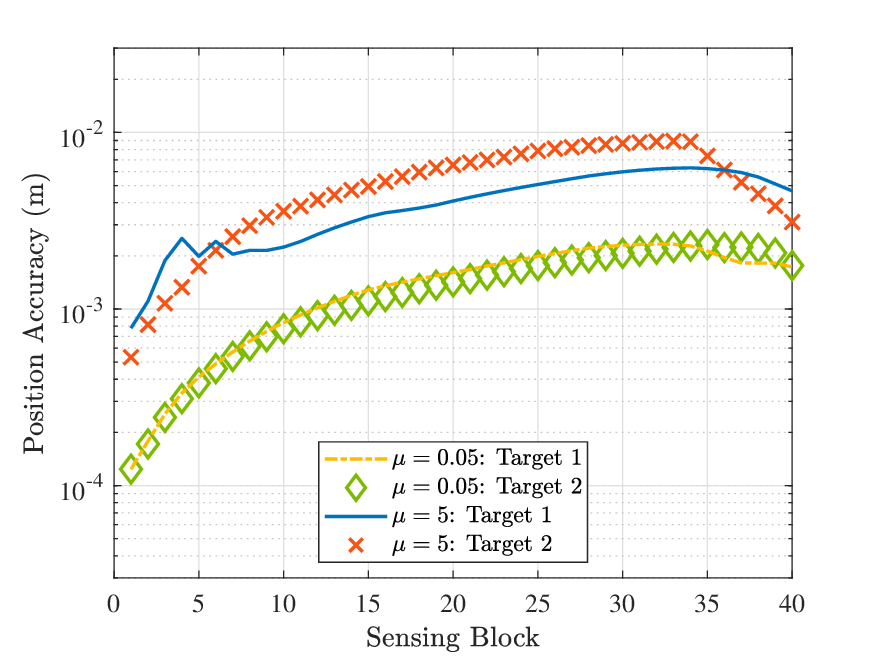}}
\subfigure[]
{\includegraphics[width=3in]{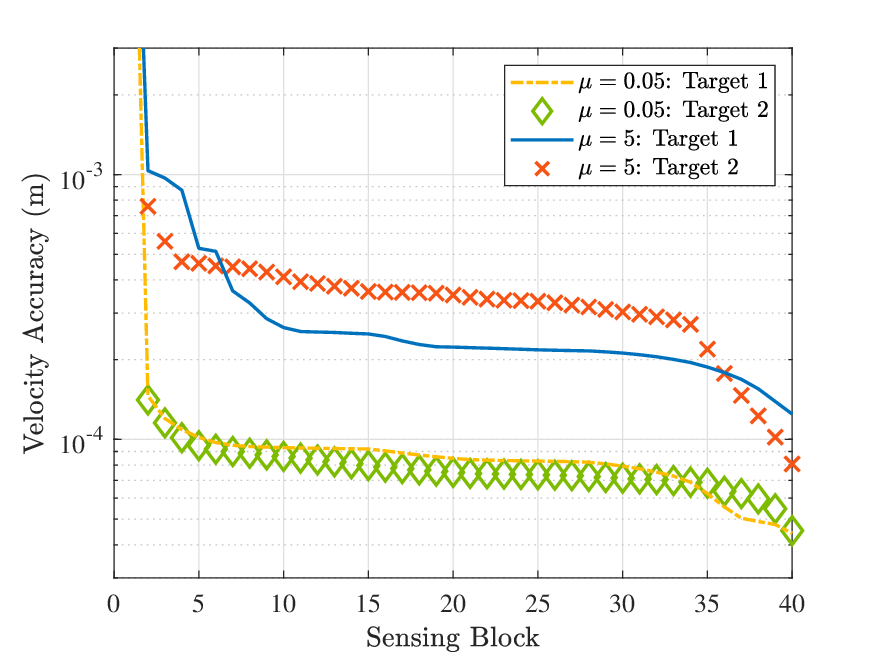}}
\caption{Tracking accuracy under different process noise power $\mu$,
with
(a) position accuracy,
and
(b) velocity accuracy.}
\label{Fig_target_RA_change}
\end{figure}

Fig. \ref{Fig_target_RA_change} shows the impact of process noise power,
i.e., $\mu$ as defined in (\ref{state_model_part_1}),
on tracking performance.
In both figures,
as the process noise $\mu$ increases,
the position and velocity errors of each target also increase.
This phenomenon occurs because
a larger $\mu$ introduces greater uncertainty in target positions.
Given that a constant velocity model is assumed,
the predicted target parameters will deviate from the actual target positions when $\mu \ne 0$.
Consequently, the reliability of the predicted information gradually diminishes over the course of tracking, leading to degraded tracking performance.
In practical applications, an accurate target motion model is able to provide more reliable information throughout the tracking process.
Therefore, the choice of the target motion model should be well-suited to the specific tracking scenario.

\section{Conclusions}\label{Section_conclusion}

This paper investigates a
BS-based networked tracking scheme
that exploits the echoes of downlink communication
signals to track multiple moving targets.
To achieve high-precision networked sensing,
we develop a networked framework
that combines a networked tracking scheme with beamforming optimization.
First,
for the tracking scheme with given beamformers,
we formulate a networked tracking model that captures both temporal and geometrical correlations
among the estimated parameters.
To fuse the data across these two dimensions,
an NKF algorithm is then employed
to recursively estimate each target's position, velocity,
and equivalent RCS at each sensing block.
Then,
we address the challenge of optimizing beamforming to enhance tracking accuracy,
despite the unavailability of future target states.
To overcome this,
we incorporate predicted target parameters into the beamforming design.
Simulation results validate the effectiveness of our approach,
demonstrating that BS beam allocation dynamically adapts to target motion,
and that the proposed tracking-based sensing method
achieves an effective target association at each BS.
Moreover,
simulation results verify that the proposed scheme
significantly outperforms conventional non-tracking-based sensing schemes.

Furthermore,
since the subcarrier allocation scheme among BSs
significantly impacts the parameter estimation accuracy
across all BSs,
optimizing this allocation
presents a promising direction for future research to enhance both estimation
and tracking performance.

\begin{appendices}

\section{Proof of Proposition \ref{proposition_1_context}}\label{proposition1_proof}

Utilizing the full FIM for subsequent derivation introduces significant computational complexity.
The following lemma is established to simplify the computation.
\begin{lemma}\label{Lemma_CRB_diagonal_matrix}
The CRB matrix $\mathbf{R}_{u,n }$
is a block-diagonal matrix containing
$K$ independent submatrices
\begin{align}
{\mathbf{R}_{u,n}} = \text{blkdiag}\left\{{\mathbf{R}_{u,n, 0, 0}}, \cdots, {\mathbf{R}_{u,n, K-1, M-1}}\right\},
\end{align}
where $\mathbf{R}_{u, n, k, m} \in \mathbb{R}^{5Q \times 5Q}$
represents the CRB matrix of the received signals of the $k$-th BS
at the $m$-th subcarrier.
\end{lemma}
\begin{proof}
The Lemma \ref{Lemma_CRB_diagonal_matrix} can be readily proved as
the received signals from different BSs are independent.
\end{proof}

Based on the definition of the CRB in (\ref{CRB_for_signal_model}) and Lemma \ref{Lemma_CRB_diagonal_matrix},
the FIM of the received signals at $k$-th BS can be partitioned into the following 25 submatrices
\begin{align}\label{doeleeknelknlkne}
{\mathbf{R}_{u,n, k, m}^{-1}}=\left[ \begin{matrix}
   {{\mathbf{J}}_{\theta \theta }} & {{\mathbf{J}}_{\theta \tau }} & {{\mathbf{J}}_{\theta f}} & {{\mathbf{J}}_{\theta {{\xi }_{R}}}} & {{\mathbf{J}}_{\theta {{\xi }_{I}}}}  \\
   \mathbf{J}_{\theta \tau }^{T} & {{\mathbf{J}}_{\tau \tau }} & {{\mathbf{J}}_{\tau f}} & {{\mathbf{J}}_{\tau {{\xi }_{R}}}} & {{\mathbf{J}}_{\tau {{\xi }_{I}}}}  \\
   \mathbf{J}_{\theta f}^{T} & \mathbf{J}_{\tau f}^{T} & {{\mathbf{J}}_{ff}} & {{\mathbf{J}}_{f{{\xi }_{R}}}} & {{\mathbf{J}}_{f{{\xi }_{I}}}}  \\
   \mathbf{J}_{\theta {{\xi }_{R}}}^{T} & \mathbf{J}_{\tau {{\xi }_{R}}}^{T} & \mathbf{J}_{f{{\xi }_{R}}}^{T} & {{\mathbf{J}}_{{{\xi }_{R}}{{\xi }_{R}}}} & {{\mathbf{J}}_{{{\xi }_{R}}{{\xi }_{I}}}}  \\
   \mathbf{J}_{\theta {{\xi }_{I}}}^{T} & \mathbf{J}_{\tau {{\xi }_{I}}}^{T} & \mathbf{J}_{f{{\xi }_{I}}}^{T} & \mathbf{J}_{{{\xi }_{R}}{{\xi }_{I}}}^{T} & {{\mathbf{J}}_{{{\xi }_{I}}{{\xi }_{I}}}}  \\
\end{matrix} \right].
\end{align}
For ease of derivation, we now define
\begin{align}
\nonumber {{{\mathbf{\tilde{Y}}}}_{n,{{k}},i,m,q}}= & \ {{{{c}}}_{n,{{k}},i,m}} {{e}^{j2\pi {{f}_{n,{{k}},q}}i{{T}_{0}}}}{{e}^{j2\pi m\Delta f{{\tau }_{n,{{k}},q}}}}\\
&\times {{{\mathbf{a}}_{r}}\left( {{\theta }_{n,{{k}},q}} \right)}\mathbf{a}_{t}^{H}\left( {{\varphi }_{n,{{k}},q}} \right),
\end{align}
so the original sensing model (\ref{formula_receiving_MIMO_all}) is equivalent to
\begin{align}\label{dlieleleel}
{{\mathbf{y}}_{n,{{k}},i,m}}=\sum\limits_{q=0}^{Q-1}{{{\xi }_{n,{{k}},q}}{{{\mathbf{\tilde{Y}}}}_{n,{{k}},i,m,q}}{{{\mathbf{{B}}}}_{n,k,m}}{{\mathbf{c}}_{n,k,i,m}}}.
\end{align}
Providing detailed derivations for all matrices would result in an overly lengthy exposition.
To maintain clarity and focus,
we instead illustrate the process using two representative submatrices
${{\mathbf{J}}_{{{\xi }_{R}}{{\xi }_{R}}}}$
and
${{\mathbf{J}}_{\theta \tau }}$
as illustrative examples.

The matrix ${{\mathbf{J}}_{{{\xi }_{R}}{{\xi }_{R}}}}$
has size $Q \times Q$,
and each element can be computed by
\begin{align}\label{dfldiejinfn}
\nonumber  &{{\left[ {{\mathbf{J}}_{{{\xi }_{R}}{{\xi }_{R}}}} \right]}_{{{q}_{0}},{{q}_{1}}}}=-\mathbb{E}\left\{ \frac{{{\partial }^{2}}\ln p\left( \left. {{\mathbf{Y}}_{n}} \right|{{\mathbf{\psi }}_{n}} \right)}{\partial \operatorname{Re}\left\{ {{\xi }_{n,k,{{q}_{0}}}} \right\}\partial \operatorname{Re}\left\{ {{\xi }_{n,k,{{q}_{1}}}} \right\}} \right\}\\
&=\frac{2}{\sigma _{z}^{2}}\mathbb{E}\left\{ \sum\limits_{i,m}{\operatorname{Re}\left\{ {{\left( \frac{\partial {{\mathbf{y}}_{n,k,i,m}}}{\partial \operatorname{Re}\left\{ {{\xi }_{n,k,{{q}_{0}}}} \right\}} \right)}^{H}}\frac{\partial {{\mathbf{y}}_{n,k,i,m}}}{\partial \operatorname{Re}\left\{ {{\xi }_{n,k,{{q}_{0}}}} \right\}} \right\}} \right\}.
\end{align}

Substituting (\ref{dlieleleel}) into (\ref{dfldiejinfn}),
we obtain the expression for the matrix ${{\mathbf{J}}_{{{\xi }_{R}}{{\xi }_{R}}}}$
as follows
\begin{align}\label{dslijdlfjdljf22}
\nonumber   & {{\left[ {{\mathbf{J}}_{{{\xi }_{R}}{{\xi }_{R}}}} \right]}_{{{q}_{0}},{{q}_{1}}}}= \frac{2}{\sigma _{z}^{2}}IM{{N}_{R}}\\
& \times \operatorname{Re}\left\{ \chi _{n,k,{{q}_{0}},{{q}_{1}}}^{\left( \xi \xi  \right)}\mathbf{a}_{t}^{T}\left( {{\varphi }_{n,k,{{q}_{0}}}} \right){\mathbf{R}_{B,n,k,m}^{*}}\mathbf{a}_{t}^{*}\left( {{\varphi }_{n,k,{{q}_{1}}}} \right) \right\},
\end{align}
where
$\chi _{n,k,{{q}_{0}},{{q}_{1}}}^{\left( \xi \xi  \right)}$
describes the correlation of the sensing signals from two targets
in the spatial-time-frequency domain,
defined as follows
\begin{align}
\nonumber & \chi _{n,k,{{q}_{0}},{{q}_{1}}}^{\left( \xi \xi  \right)}= \frac{1}{IM{{N}_{R}}}\mathbf{a}_{r}^{H}\left( {{\theta }_{n,k,{{q}_{0}}}} \right){{\mathbf{a}}_{r}}\left( {{\theta }_{n,k,{{q}_{1}}}} \right)\\
 &\times \sum\limits_{i,m}{ \gamma_{m, k} {{e}^{j2\pi \left( {{f}_{n,k,{{q}_{1}}}}-{{f}_{n,k,{{q}_{0}}}} \right)i{{T}_{0}}}}{{e}^{j2\pi m\Delta f\left( {{\tau }_{n,k,{{q}_{1}}}}-{{\tau }_{n,k,{{q}_{0}}}} \right)}}},
\end{align}
where $\gamma_{m,k}$ denotes
the indicator of the subcarrier allocation,
and $\gamma_{m,k} = 1$
if $m \in \mathcal{M}_k$,
otherwise $\gamma_{m,k} = 0$.
From (\ref{dslijdlfjdljf22}),
we can observe that the matrix
${{\mathbf{J}}_{{{\xi }_{R}}{{\xi }_{R}}}}$
is a linear function of
${\mathbf{R}_{B,n,k,m}}$.

Before proceeding with the derivation of ${{\mathbf{J}}_{\theta \tau }} \in \mathbb{R}^{Q \times Q}$,
we first provide
the partial derivatives of ${{{\mathbf{\tilde{Y}}}}_{n,k,i,m}}$
with respect to $\theta_{n, k, q}$ and $\tau_{n, k, q}$ as follows
\begin{align}\label{eeowjnldfuujk}
\nonumber  \frac{\partial {{{\mathbf{\tilde{Y}}}}_{n,k,i,m}}}{\partial {{\theta }_{n,k,q}}}& ={{c}_{n,k,i,m}}{{e}^{j2\pi {{f}_{n,k,q}}i{{T}_{0}}}}{{e}^{j2\pi m\Delta f{{\tau }_{n,k,q}}}}\\
\nonumber  &\times \left[ {{{\mathbf{\dot{a}}}}_{r}}\left( {{\theta }_{n,k,q}} \right)\mathbf{a}_{t}^{H}\left( {{\varphi }_{n,k,q}} \right)+{{\mathbf{a}}_{r}}\left( {{\theta }_{n,k,q}} \right)\mathbf{\dot{a}}_{t}^{H}\left( {{\varphi }_{n,k,q}} \right) \right] \\
 \frac{\partial {{{\mathbf{\tilde{Y}}}}_{n,k,i,m}}}{\partial {{\tau }_{n,k,q}}}& =j2\pi m\Delta f{{{\mathbf{\tilde{Y}}}}_{n,k,i,m}},
\end{align}
where
${{{\mathbf{\dot{a}}}}_{r}}\left( {{\theta }_{n,k,q}} \right)=j\frac{2\pi }{\lambda }d\cos {{\theta }_{n,k,q}}\text{diag}\left\{ {{\bm{\rho }}_{{{N}_{R}}}} \right\}{{\mathbf{a}}_{r}}\left( {{\theta }_{n,k,q}} \right)$
and
${{{\mathbf{\dot{a}}}}_{t}}\left( {{\varphi }_{n,k,q}} \right)=j\frac{2\pi }{\lambda }d\cos {{\theta }_{n,k,q}}\text{diag}\left\{ {{\bm{\rho }}_{{{N}_{T}}}} \right\}{{\mathbf{a}}_{t}}\left( {{\varphi }_{n,k,q}} \right)$
represent the
partial derivatives of the steering vectors with respect to the AOA,
and
$\bm{\rho}_{N_T} = \left[0, 1, \cdots, N_T - 1\right]^{T} \in \mathbb{Z}^{N_T \times 1}$
and
$\bm{\rho}_{N_R} = \left[0, 1, \cdots, N_R - 1\right]^{T} \in \mathbb{Z}^{N_R \times 1}$.
Using (\ref{eeowjnldfuujk}),
the expression for the matrix ${{\mathbf{J}}_{\theta \tau }}$
is given in (\ref{dslddijdl43534346fjdljf22}) at the top of the next page.
\begin{figure*}[ht]
\begin{align}\label{dslddijdl43534346fjdljf22}
\nonumber &{{\left[ {{\mathbf{J}}_{\theta \tau }} \right]}_{{{q}_{0}},{{q}_{1}}}}= -\frac{2}{\sigma _{z}^{2}}IM{{N}_{R}}\left( 2\pi \Delta f \right)\\
& \times \operatorname{Im}\left\{ \xi _{n,k,{{q}_{0}}}^{*}{{\xi }_{n,k,{{q}_{1}}}}\left[ \dot{\chi }_{n,k,{{q}_{0}},{{q}_{1}}}^{\left( \theta \tau  \right)}\mathbf{a}_{t}^{T}\left( {{\varphi }_{n,k,{{q}_{0}}}} \right)\mathbf{R}_{B,n,k,m}^{*}{{\mathbf{a}}^{*}_{t}}\left( {{\varphi }_{n,k,{{q}_{1}}}} \right)+\chi _{n,k,{{q}_{0}},{{q}_{1}}}^{\left( \theta \tau  \right)}{\mathbf{\dot{a}}_{t}^{T}\left( {{\varphi }_{n,k,{{q}_{0}}}} \right)\mathbf{R}_{B,n,k,m}^{*}{{\mathbf{a}}^{*}_{t}}\left( {{\varphi }_{n,k,{{q}_{1}}}} \right)} \right] \right\}.
\end{align}
\normalsize
\hrulefill
\vspace*{4pt}
\end{figure*}
In (\ref{dslddijdl43534346fjdljf22}),
${\dot{\chi }_{n,k,{{q}_{0}},{{q}_{1}}}^{\left( \theta \tau  \right)}}$
and
${\chi _{n,k,{{q}_{0}},{{q}_{1}}}^{\left( \theta \tau  \right)}}$
represent the spatial-frequency-weighted correlation functions,
given by
\begin{align}
\nonumber \dot{\chi }_{n,k,{{q}_{0}},{{q}_{1}}}^{\left( \theta \tau  \right)}= & \ \frac{1}{IM{{N}_{R}}}\mathbf{\dot{a}}_{r}^{H}\left( {{\theta }_{n,k,{{q}_{0}}}} \right){{\mathbf{a}}_{r}}\left( {{\theta }_{n,k,{{q}_{1}}}} \right) \\
\nonumber  & \times \sum\limits_{i,m}{m{\gamma_{m,k}}{{e}^{j2\pi \left( {{f}_{n,k,{{q}_{1}}}}-{{f}_{n,k,{{q}_{0}}}} \right)i{{T}_{0}}}}} \\
\nonumber &\times {{e}^{j2\pi m\Delta f\left( {{\tau }_{n,k,{{q}_{1}}}}-{{\tau }_{n,k,{{q}_{0}}}} \right)}}\\
\nonumber  \chi _{n,k,{{q}_{0}},{{q}_{1}}}^{\left( \theta \tau  \right)}=& \ \frac{1}{IM{{N}_{R}}}\mathbf{a}_{r}^{H}\left( {{\theta }_{n,k,{{q}_{0}}}} \right){{\mathbf{a}}_{r}}\left( {{\theta }_{n,k,{{q}_{1}}}} \right) \\
\nonumber  & \times \sum\limits_{i,m}{m{\gamma_{m,k}}{{e}^{j2\pi \left( {{f}_{n,k,{{q}_{1}}}}-{{f}_{n,k,{{q}_{0}}}} \right)i{{T}_{0}}}}}\\
 &\times {{e}^{j2\pi m\Delta f\left( {{\tau }_{n,k,{{q}_{1}}}}-{{\tau }_{n,k,{{q}_{0}}}} \right)}}.
\end{align}
It is worth noting that
${{\mathbf{J}}_{\theta \tau }}$
is also linear with respect to
${\mathbf{R}_{B,n,k,m}}$.

Other submatrices in (\ref{doeleeknelknlkne}) can also be derived using the similar
method applied for ${{\mathbf{J}}_{{{\xi }_{R}}{{\xi }_{R}}}}$
and
${{\mathbf{J}}_{\theta \tau }}$.
After obtaining all the closed-form results for these submatrices,
it can be readily concluded that the covariance matrix $\mathbf{R}_{u,n}^{-1}$ is the linear function of ${\mathbf{R}_{B,n,k,m}}$.
Then,
we examine the BIM.
According to (\ref{dfodojfdojfodjf333dfdf}),
the matrix ${{\mathbf{J}}_{P}}\left( {{\bm{\kappa }}_{n}} \right)$
is related to ${{\mathbf{R}}_{\kappa }}$, ${{\mathbf{F}}_{\kappa }}$, and $\mathbf{J}_{B}^{-1}\left( {{\bm{\kappa }}_{n-1}} \right)$.
These matrices are all irrelevant to ${\mathbf{R}_{B,n,k,m}}$ at the $n$-th block.
The data information matrix ${{\mathbf{J}}_{D}}\left( {{\bm{\kappa }}_{n}} \right)$
depends on $\mathbf{H}_n$ and $\mathbf{R}_{u,n}^{-1}$, where, as demonstrated earlier, only $\mathbf{R}_{u,n}^{-1}$ is linear with respect to ${\mathbf{R}_{B,n,k,m}}$.
Finally,
it can be concluded that
the BIM
is a linear function of
${\mathbf{R}_{B,n,k,m}}$.
Proposition \ref{proposition_1_context} is thus proved.

\section{Proof of Proposition \ref{proposition_optimization_span}}\label{Appendix_proposition_optimization_span}

Since
$\mathbf{{R}}^{\left({\left. n + 1 \right|n}\right)}_{\text{op},B, k, m}$
is a positive definite matrix,
it can be decomposed as
\begin{align}\label{sdldfdfdfkjfldkjfapp}
\mathbf{{R}}^{\left({\left. n + 1 \right|n}\right)}_{\text{op},B, k, m}={{\bm{\Delta }}^{\left({\left. n + 1 \right|n}\right)}_{k, m}}\left(\bm{\Delta }^{\left({\left. n + 1 \right|n}\right)}_{k, m}\right)^{H},
\ \forall n,k,m.
\end{align}
Define ${{\mathbf{U}}^{\left({\left. n + 1 \right|n}\right)}_{k}}=\left[ {{\mathbf{A}}^{\left({\left. n + 1 \right|n}\right)}_{t, k}},{{{\mathbf{\dot{A}}}}^{\left({\left. n + 1 \right|n}\right)}_{t, k}} \right]\in {{\mathbb{C}}^{{{N}_{T}}\times 2Q}}$,
where
${{\mathbf{A}}^{\left({\left. n + 1 \right|n}\right)}_{t, k}}=\left[ {{\mathbf{a}}_{t}}\left( {\hat{\varphi }^{\left({\left. n + 1 \right|n}\right)}_{k,0}} \right),\cdots ,{{\mathbf{a}}_{t}}\left( {\hat{\varphi }^{\left({\left. n + 1 \right|n}\right)}_{k,Q-1}} \right) \right]\in {{\mathbb{C}}^{{{N}_{T}}\times Q}}$
and
${{{\mathbf{\dot{A}}}}^{\left({\left. n + 1 \right|n}\right)}_{t, k}}=\left[ {{{\mathbf{\dot{a}}}}_{t}}\left( {\hat{\varphi }^{\left({\left. n + 1 \right|n}\right)}_{k,0}} \right),\cdots ,{{{\mathbf{\dot{a}}}}_{t}}\left( {\hat{\varphi }^{\left({\left. n + 1 \right|n}\right)}_{k,Q-1}} \right) \right]\in {{\mathbb{C}}^{{{N}_{T}}\times Q}}$
denote prediction of the $\left(n+1\right)$-th sensing block's array manifold
and
its derivative of BS $k$ at sensing block $n$, respectively.
Moreover,
let ${{\mathbf{P}}^{\left({\left. n + 1 \right|n}\right)}_{u,k}}$
and
$\mathbf{P}_{u,k}^{\bot{\left({\left. n + 1 \right|n}\right)} }$
represent
the orthogonal projection onto the subspace spanned by the columns of ${{\mathbf{U}}^{\left({\left. n + 1 \right|n}\right)}_{k}}$
and
the orthogonal complement of ${{\mathbf{P}}^{\left({\left. n + 1 \right|n}\right)}_{u,k}}$,
respectively.
Adopting above definitions,
the matrix ${{\bm{\Delta }}^{\left({\left. n + 1 \right|n}\right)}_{k, m}}$
can be further decomposed into two orthogonal combinations
\begin{align}\label{sdlkjfldkjfapp}
\nonumber &{{\bm{\Delta }}^{\left({\left. n + 1 \right|n}\right)}_{k, m}}={{\mathbf{P}}^{\left({\left. n + 1 \right|n}\right)}_{u,k}}{{\bm{\Delta }}^{\left({\left. n + 1 \right|n}\right)}_{k, m}}+\left( \mathbf{I}-{{\mathbf{P}}^{\left({\left. n + 1 \right|n}\right)}_{u,k}} \right){{\bm{\Delta }}^{\left({\left. n + 1 \right|n}\right)}_{k, m}}\\
& ={{\mathbf{P}}^{\left({\left. n + 1 \right|n}\right)}_{u,k}}{{\bm{\Delta }}^{\left({\left. n + 1 \right|n}\right)}_{k, m}}+\mathbf{P}_{u,k}^{\bot{\left({\left. n + 1 \right|n}\right)} }{{\bm{\Delta }}^{\left({\left. n + 1 \right|n}\right)}_{k, m}}, \ \forall n,k,m.
\end{align}
Substituting (\ref{sdlkjfldkjfapp})
into (\ref{sdldfdfdfkjfldkjfapp}),
we have
\begin{align}
\mathbf{{R}}^{\left({\left. n + 1 \right|n}\right)}_{\text{op},B, k, m}={{\mathbf{R}}^{\left({\left. n + 1 \right|n}\right)}_{1,k, m}}+{{\mathbf{R}}^{\left({\left. n + 1 \right|n}\right)}_{2,k, m}}, \ \forall n,k,m,
\end{align}
where
${{\mathbf{R}}^{\left({\left. n + 1 \right|n}\right)}_{1,k, m}}={{\mathbf{P}}^{\left({\left. n + 1 \right|n}\right)}_{u,k}}{{\bm{\Delta }}^{\left({\left. n + 1 \right|n}\right)}_{k, m}}\left({{\mathbf{P}}^{\left({\left. n + 1 \right|n}\right)}_{u,k}}{{\bm{\Delta }}^{\left({\left. n + 1 \right|n}\right)}_{k, m}}\right)^{H}$,
and
${{\mathbf{R}}^{\left({\left. n + 1 \right|n}\right)}_{2,k, m}}=
\mathbf{P}_{u,k}^{\bot{\left({\left. n + 1 \right|n}\right)} }{{\bm{\Delta }}^{\left({\left. n + 1 \right|n}\right)}_{k, m}}\left({{\mathbf{P}}^{\left({\left. n + 1 \right|n}\right)}_{u,k}}{{\bm{\Delta }}^{\left({\left. n + 1 \right|n}\right)}_{k, m}}\right)^{H}
+{{\mathbf{P}}^{\left({\left. n + 1 \right|n}\right)}_{u,k}}{{\bm{\Delta }}^{\left({\left. n + 1 \right|n}\right)}_{k, m}}\left(\mathbf{P}_{u,k}^{\bot{\left({\left. n + 1 \right|n}\right)} }{{\bm{\Delta }}^{\left({\left. n + 1 \right|n}\right)}_{k, m}}\right)^{H}
+\mathbf{P}_{u,k}^{\bot{\left({\left. n + 1 \right|n}\right)} }{{\bm{\Delta }}^{\left({\left. n + 1 \right|n}\right)}_{k, m}}\left(\mathbf{P}_{u,k}^{\bot{\left({\left. n + 1 \right|n}\right)} }{{\bm{\Delta }}^{\left({\left. n + 1 \right|n}\right)}_{k, m}}\right)^{H}$.

Since
${{\mathbf{R}}^{\left({\left. n + 1 \right|n}\right)}_{2,k, m}}$
is the matrix orthogonal to ${{\mathbf{U}}^{\left({\left. n + 1 \right|n}\right)}_{k}}$,
it can be readily proved that
the property
$\left({{\mathbf{A}}^{\left({\left. n + 1 \right|n}\right)}_{t, k}}\right)^H{{\mathbf{R}}^{\left({\left. n + 1 \right|n}\right)}_{2,k, m}}{{\mathbf{A}}^{\left({\left. n + 1 \right|n}\right)}_{t, k}}
=\left({{\mathbf{\dot{A}}}^{\left({\left. n + 1 \right|n}\right)}_{t, k}}\right)^H{{\mathbf{R}}^{\left({\left. n + 1 \right|n}\right)}_{2,k, m}}{{\mathbf{A}}^{\left({\left. n + 1 \right|n}\right)}_{t, k}}
=\left({{\mathbf{A}}^{\left({\left. n + 1 \right|n}\right)}_{t, k}}\right)^H{{\mathbf{R}}^{\left({\left. n + 1 \right|n}\right)}_{2,k, m}}{{\mathbf{\dot{A}}}^{\left({\left. n + 1 \right|n}\right)}_{t, k}}
=\left({{\mathbf{\dot{A}}}^{\left({\left. n + 1 \right|n}\right)}_{t, k}}\right)^H{{\mathbf{R}}^{\left({\left. n + 1 \right|n}\right)}_{2,k, m}}{{\mathbf{\dot{A}}}^{\left({\left. n + 1 \right|n}\right)}_{t, k}}=\mathbf{0}$
holds.
Therefore, the above results imply that
\begin{align}\label{dlfkjdlkjflkd3rwr2r2}
\nonumber &{{\left. {{\mathbf{J}}_{B}}\left( {{\bm{\hat{\kappa }}}_{\left. n + 1 \right|n}} \right) \right|}_{{{\mathbf{R}}_{B}}={{\mathbf{R}}^{\left({\left. n + 1 \right|n}\right)}_{2,k, m}}}}=\mathbf{0},\\
&\text{tr}\left\{ {{\mathbf{R}}^{\left({\left. n + 1 \right|n}\right)}_{2,k, m}} \right\}=\left\| \mathbf{P}_{u,k}^{\bot{\left({\left. n + 1 \right|n}\right)} }{{\bm{\Delta }}^{\left({\left. n + 1 \right|n}\right)}_{k}} \right\|_{F}^{2}\ge 0.
\end{align}
Due to the power constraint
$\text{tr}\left\{ \mathbf{{R}}^{\left({\left. n + 1 \right|n}\right)}_{\text{op},B, k, m} \right\}=\text{tr}\left\{ {{\mathbf{R}}^{\left({\left. n + 1 \right|n}\right)}_{1,k, m}} \right\}+\text{tr}\left\{ {{\mathbf{R}}^{\left({\left. n + 1 \right|n}\right)}_{2,k, m}} \right\}={{P}_{k,m}}$,
we can conclude that
$\text{tr}\left\{ {{\mathbf{R}}^{\left({\left. n + 1 \right|n}\right)}_{1,k,m}} \right\}\le {{P}_{k,m}}$.
Moreover, note that
for a given beamformer,
increasing the power will enhance the tracking performance.
Therefore,
for optimal beamformer,
more transmit power should be allocated to ${{\mathbf{R}}^{\left({\left. n + 1 \right|n}\right)}_{1,k,m}}$,
due to the BIM relationship derived in
(\ref{dlfkjdlkjflkd3rwr2r2}),
to minimize the PCRB.
This implies that
$\mathbf{P}_{u,k}^{\bot{\left({\left. n + 1 \right|n}\right)} }{{\bm{\Delta }}^{\left({\left. n + 1 \right|n}\right)}_{k,m}}=\mathbf{0}$
and
$\text{tr}\left\{ {{\mathbf{R}}^{\left({\left. n + 1 \right|n}\right)}_{1,k,m}} \right\}={{P}_{k,m}}$.
Hence, we finally obtain that $\mathbf{{R}}^{\left({\left. n + 1 \right|n}\right)}_{\text{op},B, k,m}={{\mathbf{R}}^{\left({\left. n + 1 \right|n}\right)}_{1,k,m}}$,
which means
$\mathbf{{R}}^{\left({\left. n + 1 \right|n}\right)}_{\text{op},B, k,m}$
is spanned by ${{\mathbf{{A}}}^{\left({\left. n + 1 \right|n}\right)}_{t, k}}$ and ${{\mathbf{\dot{A}}}^{\left({\left. n + 1 \right|n}\right)}_{t, k}}$.
Proposition \ref{proposition_optimization_span} is thus proved.

\end{appendices}

\bibliographystyle{IEEEtran}
\bibliography{AAA}
\vfill

\end{document}